\documentclass[superscriptaddress, twocolumn, aps,nofootinbib]{revtex4-1}

\newcommand{\eqdef}{\coloneqq}

\usepackage{bbm}
\usepackage{amssymb}
\usepackage{graphicx}
\usepackage{graphics}
\usepackage{afterpage}
\usepackage[abs]{overpic}
\usepackage{amsmath}
\usepackage{braket}
\usepackage{amsthm}
\usepackage{color}
\usepackage{dsfont}
\usepackage{appendix}
\usepackage{xcolor}
\usepackage{ucs}
\usepackage{graphicx}
\usepackage{url}

\usepackage{xcolor}
\usepackage{lipsum}

\usepackage[ colorlinks = true,              linkcolor = blue, urlcolor  =
blue,              citecolor = red,              anchorcolor = green,
]{hyperref}%

\usepackage{mathtools}

\def\tA{\tilde{A}}
\def\tB{\tilde{B}}

\def\p{\mathbf{p}}
\def\q{\mathbf{q}}

\def\u{\mathbf{u}}

\def\H{\mathbf{H}}
\def\D{\mathbf{D}}

\newtheorem{theorem}{Theorem}
\newtheorem*{theorem*}{Theorem}

\newtheorem{lemma}{Lemma}
\newtheorem*{lemma*}{Lemma}

\newtheorem{definition}{Definition}
\newtheorem*{definition*}{Definition}

\allowdisplaybreaks

\begin{document} 
\title{Quantum conditional entropy from information-theoretic principles
}

\author{Sarah \surname{Brandsen}}\email{sarah.brandsen@duke.edu}
\address{
Department of Physics,
Duke University, 
Durham, NC, USA 27708}
\address{
Department of Mathematics and Statistics, Institute for Quantum Science and Technology,
University of Calgary, AB, Canada T2N 1N4}
\author{Isabelle Jianing \surname{Geng}}
\address{
Department of Mathematics and Statistics, Institute for Quantum Science and Technology,
University of Calgary, AB, Canada T2N 1N4}
\author{Mark M.~\surname{Wilde}}
\address{
Hearne Institute for Theoretical Physics, Department of Physics and Astronomy, and Center for Computation and Technology,
Louisiana State University, Baton Rouge, Louisiana 70803, USA}
\author{Gilad \surname{Gour}}
\address{
Department of Mathematics and Statistics, Institute for Quantum Science and Technology,
University of Calgary, AB, Canada T2N 1N4}

\begin{abstract}
We introduce an axiomatic approach for characterizing quantum conditional entropy. Our approach relies on two physically motivated axioms: monotonicity under conditional majorization and additivity. We show that these two axioms provide sufficient structure that enable us to derive several key properties applicable to all quantum conditional entropies studied in the literature. Specifically, we prove that any quantum conditional entropy \emph{must} be negative on certain entangled states and must equal $-\log d$ on $d\times d$ maximally entangled states.  We also prove the non-negativity of conditional entropy on separable states, and we provide a generic definition for the dual of a quantum conditional entropy. Finally, we develop an operational approach for characterizing quantum conditional entropy via games of chance, and we show that, for the classical case, this complementary approach yields the same ordering as the axiomatic approach. 


\end{abstract}
\maketitle

\section{Introduction}

Conditional entropy quantifies the uncertainty of a bipartite state shared between Alice and Bob, given access to Bob's system only. Previous work on conditional entropy for quantum bipartite states has demonstrated that, unlike classical conditional entropy, quantum conditional entropy can be negative~\cite{Cerf_1997} and achieves its minimum on maximally entangled states. The initial discovery that quantum conditional entropy could be negative was so surprising that various conjectures about its meaning were put forth by the authors of \cite{Cerf_1997}, such as its potential relation to ``anti-qubits,'' which are quanta of negative information. The most satisfying information theoretic interpretation of quantum conditional entropy and its negativity is given by the state merging protocol of \cite{Horodecki2005}, which demonstrates that conditional entropy is the rate at which maximally entangled Bell states need to be consumed in order transfer Alice's share of a bipartite state to Bob, who possesses the other share of the state. The importance of the existence of negative conditional entropy is also reflected in the fact that there is an information quantity, called coherent information \cite{PhysRevA.54.2629}, dedicated to measuring how negative the conditional entropy of a bipartite quantum state is.
Multiple previous works have investigated topics related to conditional entropy~\cite{Liu_2018, Patro_2017, Capel_2018, Gour5, Kuznetsova2010QuantumCE, doi:10.1063/1.5027495, PhysRevA.99.062119, Gour, PhysRevA.104.012417} and have found additional applications of conditional entropy to quantum cryptographic protocols~\cite{Brown_2021}. 

The central role of conditional entropy in quantum information theory motivates the need for an axiomatic approach that defines the set of all possible quantum conditional entropies. Previous literature has extensively studied axiomatic derivations for other entropies~\cite{e13111945, e10030261}, beginning with approaches deriving the Shannon entropy~\cite{6773024} and progressing to a complete axiomatic approach for classical entropies and relative entropies~\cite{gour2021entropy}. More recently, an axiomatic approach for entropy based on core-concavity was developed and connections were made between such core-concave functions and classical conditional entropy~\cite{9064819}.

In this work, we start with a minimal set of assumptions with strong information-theoretic motivation and demonstrate that these simple assumptions are sufficient for recovering key properties of quantum conditional entropy, such as its negativity for maximally entangled states, non-negativity for separable bipartite states, invariance under local unitaries, and reduction to an entropy function on uncorrelated systems. 
At the core of our approach, we fully characterise the set of channels that are conditional entropy non-decreasing, and we require that conditional entropy be monotonic under the action of such channels. More specifically, we require that any conditional entropy non-decreasing channel maps a state of maximal entropy on Alice's side to an output state that also has maximal entropy on Alice's side, a condition which is satisfied by \emph{conditionally unital} channels. Likewise, given that conditional entropy can be understood as measuring as how much information Bob has about Alice's system, any conditional entropy non-decreasing channel cannot ``leak'' information from Alice's system to Bob's system. Thus, conditional entropy non-decreasing channels must also satisfy the semi-causality condition outlined in~\cite{Beckman_2001}. In other words, we identify the set of conditional entropy non-decreasing channels with channels that are both conditionally unital and semi-causal (CUSC).

This gives rise to a natural majorization relation where $\sigma_{AB} \precsim_{A} \rho_{AB}$ (i.e. $\rho_{AB}$ has less uncertainty than $\sigma_{AB}$ with respect to system $A$) if $\sigma_{AB}$ can be obtained by applying a CUSC channel to $\rho_{AB}$. Any conditional entropy is thus required to be monotonic under this partial ordering which we call \emph{conditional majorization}. 
In~\cite{Gour} a different definition of conditional majorization (restricted to classical systems) was given, however, we prove that the quantum conditional majorization defined in this paper yields the exact same ordering as in~\cite{Gour} on (bipartite) classical systems.

An alternative, operationally motivated means for characterizing uncertainty is to utilize games of chance. It was found in~\cite{brandsen2021entropy} that games of chance provide an operational interpretation to conditional majorization for classical bipartite states.  Games of chance are ideal for studying uncertainty, as the probability of winning a game with a given physical system depends solely on the uncertainty of the system's output. Such games of chance thus lead to a natural partial ordering between channels, where state a $\rho_{AB}$ is ``less uncertain'' than state $\sigma_{AB}$ if $\rho_{AB}$ performs at least as well as $\rho_{AB}$ for any game of chance. In this work, we extend previous results by providing games of chance capable of characterising uncertainty of quantum bipartite states.

Lastly, we provide an alternative approach for characterising the entropy of a quantum channel~\cite{Gour_2021}, which is operationally motivated by games of chance. Our results reflect that entanglement is a resource~\cite{Chitambar_2019, Sparaciari2020firstlawofgeneral}, as entanglement-preserving unitary channels strictly outperform entanglement-breaking channels. We additionally compute the reward function for key quantum gambling games with several special classes of quantum channels, such as unitary, amplitude damping, depolarizing, measurement, and dephasing channels.

\section{Notation}

We denote quantum channels in uppercase calligraphic letters, where $\mathcal{N} \in \operatorname{CPTP}\left(AB \rightarrow {A'B'}\right)$ denotes a completely positive, trace preserving map (i.e. quantum channel) that takes as an input a bipartite state of systems $A$ and $B$ and outputs a bipartite state of systems $A'$ and $B'$. At times, we use subscripts to indicate the input and output system of a channel such as $\mathcal{N}_{{AB}\to{A'B'}}$. The dimension of a system is denoted as $|A|$, $|B|$, e.t.c.. We also use the tilde symbol to indicate a replica of a system. For example, $\tA$ or $\tB$ represents copies of $A$ and $B$, and in particular $|A|=|\tA|$ and $|B|=|\tB|$.

Likewise, we denote a classical channel by $\mathcal{T}$, whose corresponding transition matrix is a column stochastic matrix $T = (t_{w|z'} )_{w,z'}$ such that each entry $t_{w|z'}\in[0,1]$ and  $\sum_{w}t_{w|z'}=1$ for each $z'$.

In the case in which the bipartite state has a trivial component (e.g., $|B| = |B'| = 1$), we omit the letters that refer to the trivial system and denote the channel $\mathcal{N}$ as $\mathcal{N} \in \operatorname{CPTP}\left(A \rightarrow A'\right)$. Throughout this work, the computational basis of a given system with dimension $d$ is denoted as $\{\ket{1},\ldots,\ket{d}\}$. We denote the maximally entangled Bell state for systems $A$ and $B$ (where $|A| = |B|$) as
$$\ket{\phi_{AB}^{+}} \triangleq \frac{1}{\sqrt{|A|}} \sum_{j=1}^{|A|} \ket{jj}\ ,$$ where $\ket{jj}\triangleq\ket{j}\otimes\ket{j}$,
and denote the corresponding unnormalized state as
$$\ket{\Phi_{AB}} \triangleq \sum_{j=1}^{|A|} \ket{jj}.$$

For a matrix $M$, the sum of its $w$ largest singular values will be denoted as the Ky-fan $w$ norm $\|M\|_{\left(w\right)}$.
We denote with $I_{A}$ the identity operator on system $A$, and denote with $\mathbf{u}_{A} = \frac{1}{|A|} I_{A}$ the maximally mixed state in system $A$. When the systems on which the operators act are clear from the context, we might omit the subscripts indicating the systems.  

The Choi matrix $J^{\mathcal{N}}_{AB}$ for a given channel $\mathcal{N}_{A \rightarrow B}$ is defined by the action of $\mathcal{N}$ on the maximally entangled operator as
\begin{align}
    J^{\mathcal{N}}_{AB} = \sum_{i, j =1}^{|A|} \ket{i}\!\!\bra{j} \otimes \mathcal{N}(\ket{i}\!\!\bra{j} ).
\end{align}

Quantum measurements are denoted as $\hat{\Pi} = \{\Pi^{(j)}\}_{j}$ where $\sum_{j} \Pi^{(j)} = I$ and for every $j$, $0 \leq \Pi^{(j)} \leq I$. At times, we use subscript to denote the system on which the measurements act, e.g., $\hat{\Pi}_{A}$ is a measurements on system $A$. Unless otherwise specified, we consider measurements which have rank-one projective elements. Finally, the set of density matrices corresponding to system $A$ is denoted as $\mathfrak{D}\left(A\right)$, and $\u_A$ denotes the uniform (i.e. maximally mixed) state in $\mathfrak{D}(A)$.

\section{Axiomatic Approach to Quantum Conditional Entropy}

We begin by completely characterising the set of bipartite quantum channels which are conditional entropy non-decreasing. A bipartite quantum channel $\mathcal{N}$ is conditional entropy non-decreasing if for any bipartite state $\rho_{AB}$, the state $\mathcal{N}(\rho_{AB})$ has at least as much conditional entropy. 

Since conditional entropy measures the uncertainty of system $A$ when one is given access to system $B$, evidently allowing information to leak from system $A$ to system $B$ could decrease conditional entropy. This gives rise to the semi-causality requirement, which is defined as follows

\begin{definition}
A bipartite channel $\mathcal{N} _{AB \rightarrow AB'}$ is $A\not\to B'$ semi-causal if any channel $\mathcal{M} _{A \rightarrow A}$ that Alice applies to her system cannot be detected by Bob. Formally, this requirement can be stated as follows:
\begin{align}
    \mathcal{N} _{AB \rightarrow B'} \circ \mathcal{M} _{A \rightarrow A} = \mathcal{N} _{AB \rightarrow B'},
\end{align}
for all $\mathcal{M} \in \operatorname{CPTP}\left(A \rightarrow A\right)$ and where $\mathcal{N} _{AB \rightarrow B'}$ indicates that the channel output in system $A$ is traced out.
\end{definition}
See Figure~\ref{fig_semi-causal} for a depiction of the semi-causality requirement. 

\begin{figure}[h!]
  \begin{overpic}[scale=0.5]{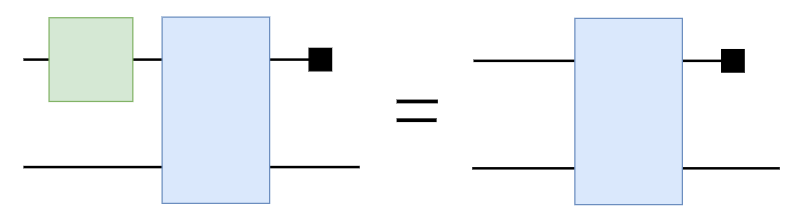}
  \put(20, 44){$\mathcal{M}$}
\put(59, 30){$\mathcal{N}$}
\put(183, 30){$\mathcal{N}$}
\end{overpic}
\caption{Depiction of semi-causal channel $\mathcal{M}$, where the black square represents the discarding channel.}
\label{fig_semi-causal}
\end{figure}
We note that any semi-causal channel can be written as a local channel on $B$ which feeds an output (corresponding to a reference system $R$) into a channel acting only on $A$ and $R$ (see Figure~\ref{fig_semi-casual_channel}). As such, information can only flow from $B$ to $A$, but never leak from $A$ to $B$.
\begin{lemma}
If $\mathcal{N} _{AB \rightarrow AB'}$ is an $A \not\to B'$ semi-causal channel, then there exists a reference system $R$, a quantum channel $\mathcal{E} \in \operatorname{CPTP}\left(AR \rightarrow A\right)$, and an isometry $\mathcal{F} \in \operatorname{CPTP}\left(B \rightarrow RB'\right)$ such that
\begin{align}
    \mathcal{N} _{AB \rightarrow AB'}\left(\rho_{AB}\right) = 
\operatorname{Tr}_{R}\Big[ \mathcal{E}_{RA \rightarrow A} \circ \mathcal{F}_{B \rightarrow RB'}\left(\rho_{AB}\right) \Big]
\end{align}
for every state $\rho_{AB}$. 
\end{lemma}

\begin{proof}
Refer to~\cite{Piani_2006} for a proof. 
\end{proof}

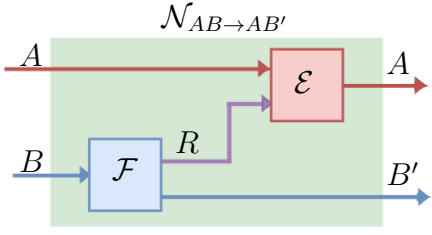
\begin{figure}
\large
  \begin{overpic}[scale=0.6]{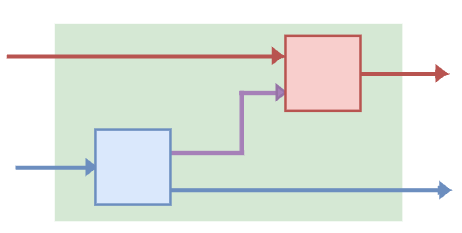}
  \put(62, 80){$\mathcal{N}_{AB\to AB'}$}
\put(43, 22){$\mathcal{F}$}
\put(112, 55){$\mathcal{E}$}
\put(67,32){$R$}
\put(8,25){$B$}
\put(8,65){$A$}
\put(147,62){$A$}
\put(147,20){$B'$}
\end{overpic}
\caption{A semi-causal bipartite channel $\mathcal{N}_{AB \rightarrow AB'}$.}
\label{fig_semi-casual_channel}
\end{figure}

Finally, we reframe the definition of semi-causality in terms of Choi matrices
\begin{lemma}
Let $\mathcal{N} \in \text{CPTP}(AB \rightarrow A B')$. If $\mathcal{N}_{AB \rightarrow \tilde{A}B'}$ is $A\not\to B'$ semi-causal, then its Choi matrix $J^{\mathcal{N}}_{AB\tA B'}$ satisfies
\begin{align}
    J^{\mathcal{N}}_{ABB'} = \mathbf{u}_{A} \otimes J^{\mathcal{N}}_{BB'}
\end{align}
\end{lemma}
\begin{proof}
Let $\mathcal{R}_{A' \rightarrow A}$ be the completely randomising channel which maps any input to the maximally mixed state $\mathbf{u}_{A}$. Since $\mathcal{N}$ is semi-causal, then 
\begin{align}
\text{Tr}_{\tilde{A}}\left[\mathcal{N}_{AB \rightarrow \tilde{A}B'} \circ \mathcal{R}_{A \rightarrow A} \right]= \text{Tr}_{\tilde{A}} \left[\mathcal{N}_{AB \rightarrow \tilde{A} B'}\right]
\end{align}
and $J^{\mathcal{N} \circ \mathcal{R}}_{AB B'} = J^{\mathcal{N}}_{ABB'}$. We note that the Choi matrix for $ \mathcal{N}_{AB \rightarrow B'} \circ \mathcal{R}$ is then equivalent to the Choi matrix for $\mathcal{N}_{AB \rightarrow B'}$ and find
\begin{align}
    & J^{\mathcal{N}}_{ABB'}\notag \\
    &= \sum_{i, j=1}^{|A|} \sum_{k, l=1}^{|B|} \ket{ik}\!\!\bra{jl}_{AB} \otimes  \mathcal{N}_{AB \rightarrow B'} \left( \mathcal{R}\left(\ket{i}\!\!\bra{j}\right) \otimes \ket{k}\!\!\bra{l}\right) \notag \\
    &= \mathbf{u}_{A} \otimes \sum_{k,l=1}^{|B|} \ket{k}\!\!\bra{l}_{B} \otimes  \mathcal{N}_{AB \rightarrow B'} \left( I_{A} \otimes \ket{k}\!\!\bra{l}\right) \notag \\
    &= \mathbf{u}_{A} \otimes M 
\end{align}
where $M = \sum_{k,l=1}^{|B|} \ket{k}\!\!\bra{l}_{B} \otimes  \mathcal{N}_{AB \rightarrow B'} \left( I_{A} \otimes \ket{k}\!\!\bra{l}\right)$. Upon taking the partial trace of $A$ over both sides, then $M = J^{\mathcal{N}}_{BB'}$ and the theorem follows. 
\end{proof}

In addition to semi-causality, any entropy non-decreasing channel must map states of maximal conditional entropy (i.e. states of the form $\mathbf{u}_{A} \otimes \rho_{B}$) to other states of maximal conditional  entropy. We define channels satisfying this condition to be conditionally unital 
\begin{definition}
A bipartite channel $\mathcal{N} _{AB \rightarrow {AB'}}$ is a conditional unital channel if for all $\rho \in \mathfrak{D}\left(B\right)$, there exists a state $\sigma_{B'} \in \mathfrak{D}\left(B'\right)$ such that
\begin{align}
    \mathcal{N}\left(\mathbf{u}_{A} \otimes \rho_{B}\right) = \mathbf{u}_{A} \otimes \sigma_{B'}.
\end{align}
\end{definition}

We can also reframe the definition of a conditional unital channel in terms of Choi matrices. 
\begin{lemma}
$\mathcal{N} \in \operatorname{CPTP}(AB \rightarrow \tilde{A}B')$ is conditional unital if and only if its Choi matrix $J^{\mathcal{N}}_{AB\tilde{A}B'}$ satisfies 
\begin{align}
    J^{\mathcal{N}}_{B \tilde{A} B'} = J^{\mathcal{N}}_{BB'} \otimes \mathbf{u}_{\tilde{A}}.
\end{align}
\end{lemma} 

\begin{proof}
Recall that the definition of a conditional unital channel can be written as
\begin{align}
    & \mathcal{N}\left(\mathbf{u} _{A} \otimes \rho_{B}\right) \notag \\
    &= \operatorname{Tr}_{AB}\Big[ J^{\mathcal{N}}_ {AB\tilde{A}B'} \left( \mathbf{u}_{A} \otimes \left(\rho_{B}\right)^{T} \otimes I_{\tilde{A}B'} \right)\Big] \\
    &= \frac{1}{|A|} \operatorname{Tr}_{B}\Big[ J^{\mathcal{N}}_{B\tilde{A}B'} \left( \left(\rho_{B}\right)^{T} \otimes I_{\tilde{A}B'} \right)\Big] \\
    &\triangleq \mathbf{u}_{A} \otimes \sigma_{B'}.
\end{align}
We now multiply both sides by $|A|\eta_{\tilde{A}} \otimes \omega_{B'}$ where $\eta_{\tilde{A}}$ is traceless and Hermitian and $\omega_{B'}$ is a Hermitian matrix. It follows that
\begin{align}
 0 &=  \operatorname{Tr}\Big[ \left(I_{A} \otimes \sigma_{B'}\right)\left(\eta_{\tilde{A}} \otimes \omega_{B'}\right)\Big]  \\
    & = \operatorname{Tr}_{\tilde{A}B'}\Big[\operatorname{Tr}_{B}\Big[ J^{\mathcal{N}}_{B\tilde{A}B'} \left( \rho_{B}^{T} \otimes I_{\tilde{A}B'} \right)\Big] \left( \eta_{\tilde{A}} \otimes \omega_{B'} \right) \Big] \\
  & = \operatorname{Tr}\Big[ J^{\mathcal{N}}_{B\tilde{A}B'} \left( \rho_{B}^{T} \otimes  \eta_{\tilde{A}} \otimes \omega_{B'} \right) \Big] .
\end{align}
From this, it follows that for any orthonormal basis $\{\rho_{B}^{(i)}\}_i$ of Hermitian operators, any orthonormal basis $\{\eta_{\tilde{A}}^{\left(j\right)} \}_j$ of traceless Hermitian operators, and any orthonormal basis $\{\omega_{B'}^{(k)}\}_k$ of Hermitian operators, then 
\begin{align}
    \mathcal{S} = \Big\{ \rho_{B}^{(i)} \otimes \eta_{\tilde{A}}^{\left(j\right)} \otimes \omega_{B'}^{\left(k\right)}\Big\}_{i, j, k}
\end{align}
spans the orthogonal complement to $J^{\mathcal{N}}_{B \tilde{A} B'}$. 

We note that $\mathbf{u}_{\tilde{A}}$ is the only matrix (up to a scalar) which is orthogonal to all traceless Hermitian matrices $\{\eta_{\left(\tilde{A}\right)}^{\left(j\right)} \}$. Thus, $J^{\mathcal{N}}_{B\tilde{A}B'}$ is spanned by $\{\rho_{B}^{\left(i\right)} \otimes \mathbf{u}_{\tilde{A}} \otimes \omega_{B'}^{\left(k\right)} \}$, so 
\begin{align}
    J^{\mathcal{N}}_{B \tilde{A} B'} &= \sum_{i, j, k} c_{i, j, k} \rho_{B}^{\left(i\right)} \otimes \omega_{B'}^{\left(k\right)} \otimes \mathbf{u}_{\tilde{A}} \\
    &= J^{\mathcal{N}}_{BB'} \otimes \mathbf{u}_{\tilde{A}}
\end{align}
This concludes the proof.
\end{proof}

Having motivated the need for a conditional entropy non-decreasing channel to be semi-causal and conditional unital, we denote the set of all conditional unital, $A\not\to B$ semi-causal channels in $\operatorname{CPTP}\left(AB \rightarrow AB'\right)$ as $\operatorname{CUSC}\left(AB \rightarrow AB'\right)$. Then bipartite state $\rho_{AB}$ is less uncertain than state $\sigma_{AB}$ if one can obtain $\sigma_{AB}$ using $\rho_{AB}$ and a CUSC channel. This gives rise to the following pre-order.


\begin{definition}[Conditional Majorization]
\label{majorisation}
Consider two bipartite states $\rho_{AB}$ and $\sigma _{A'B'}$. We say that $\rho_{AB}$ conditionally majorizes $\sigma _{A'B'}$ with respect to system $A$, and write 
\begin{equation}
    \rho_{AB}\succsim_A\sigma_{A'B'}
\end{equation}
if either $|A'|\geq |A|$ and there exists an isometry $\mathcal{V}\in{\rm CPTP}(A\to A')$, and $\mathcal{N} \in \operatorname{CUSC}\left(AB \rightarrow AB'\right)$, such that 
\begin{align}
    \sigma _{A'B'} =\mathcal{V}_{A\to A'}\circ \mathcal{N} _{AB \rightarrow AB'}\left(\rho_{AB}\right)\;,
\end{align}
or $|A|\geq |A'|$ and there exists an isometry $\mathcal{U}\in{\rm CPTP}(A'\to A)$, and $\mathcal{N} \in \operatorname{CUSC}\left(AB \rightarrow AB'\right)$, such that
\begin{align}
    \mathcal{U}_{A'\to A}\left(\sigma _{A'B'}\right) = \mathcal{N} _{AB \rightarrow AB'}\left(\rho_{AB}\right).
\end{align}
\end{definition}

The isometries introduced in the definition above enable us to compare between bipartite states of different dimensions. Note that also majorization between two probability vectors of the same dimension can be extended to vectors of different dimensions in this way (by adding zero components to the vector with the smaller dimension so that the vectors become of equal dimension). Such extension is motivated by the fact that entropy functions don't change under such embedding (see further discussion of it in~\cite{gour2021entropy}).
We are now ready to define conditional entropy based on the pre-order induced by conditional majorization. 
\begin{definition}
A nonzero function $\H : \bigcup_{A, B} \mathfrak{D}\left(AB\right) \rightarrow \mathbb{R}$ is a quantum conditional entropy if it satisfies the following monotonicity and additivity constraints
\begin{enumerate}
    \item For all $\sigma_{A'B'}$, $\rho_{AB}$, if $\sigma_{A'B'} \precsim_{A} \rho_{AB}$, then $\H(A|B)_{\rho} \leq \H\left(A'|B'\right)_{\sigma}$.
    \item $\H\left(AA'|BB'\right)_{\rho \otimes \tau} = \H(A|B)_{\rho} + \H\left(A'|B'\right)_{\tau}$ for all states $\rho_{AB}$ and $\tau_{A'B'}$.
\end{enumerate}
\end{definition}

Note that a conditional entropy $\H$ is defined on all density matrices in all finite dimensions. In particular, if $\rho\in\mathfrak{D}(AB)$ and if $|B|=1$ then we write
\begin{equation}
    \H(A|B)_\rho=\H(A)_\rho\;,
\end{equation}
where $\H(A)_\rho$ is an entropy function; that is, it is additive under tensor products and monotonic under mixing (i.e. majorization).

If a given function $\H$ satisfies the above conditions, then $\alpha \H$ is also a conditional entropy for any $\alpha > 0$. To eliminate this extra degree of freedom due to scaling, we consider only normalised conditional entropy functions such that when $|A|=2$ and $|B|=1$ (i.e. the bipartite state is a single qubit), then
    \begin{align}
        \H(A|B)_{\frac{I}{2}} = H\!\left(\frac{I}{2}\right) = 1.
    \end{align}
where $H$ is the Shannon entropy. The above three axioms (i.e. including the normalization condition) are sufficient to guarantee several key properties of conditional entropy. We start by showing that conditional entropies reduce to entropies  on uncorrelated states.

\begin{lemma}
Conditional entropy reduces to an entropy function on uncorrelated states. That is, for any product state $\omega_{A} \otimes \tau_{B}$, 
\begin{equation}
    \H(A|B)_{ \omega_A\otimes\tau_B}= \H(A)_{\omega}.
\end{equation}
\end{lemma}

\begin{proof}
Suppose that $\rho_{AB}=w_A\otimes\tau_B$ is a product state. Let $\mathcal{E}_{B\to B}\in\operatorname{CPTP}\left(B\to B\right)$ be the completely randomizing channel on system $B$. Then it follows that $\mathcal{I}_{A\to A}\otimes\mathcal{E}_{B\to B}\in\operatorname{CUSC}\left(AB\to AB\right)$ and that
\begin{equation}
    \mathcal{I}_{A\to A}\otimes\mathcal{E}_{B\to B}\left(w_A\otimes\tau_B\right)=w_{A}\otimes\u_{B} .
\end{equation}
Moreover, let $\mathcal{R}_{\tau}^B$ be the replacement channel on system~$B$, which takes any quantum state to $\tau_B$. Then 
\begin{equation}
     \mathcal{I}_{A\to A}\otimes\mathcal{R}_{B\to B}\left(w_A\otimes\u_B\right)=w_{A}\otimes\tau_{B}\ .
\end{equation}
Thus, 
\begin{equation}
    w_A\otimes\tau_B\succsim_{A}w_A\otimes\u_B\quad\text{and}\quad w_A\otimes\u_B\succsim_{A}w_A\otimes\tau_B
\end{equation}
which implies, by the monotonicity property of the conditional entropy, that
\begin{equation}
    \H(A|B)_{ w_A\otimes\tau_B}= \H(A|B)_{ w_A\otimes\u_B}.
\end{equation}
Thus, $\H(A|B)_{ w_A\otimes\tau_B}$ depends on $w_A$ only. Moreover, the function $w_A\mapsto \H(A|B)_{ w_A\otimes\tau_B}$ satisfies the three axioms of entropy \cite{Gour5} and therefore can be considered as an entropy of $w_A$.
\end{proof}

The same statement also holds when system $|B|$ is trivial, i.e., when $|B|=1$. 

Unlike entropy, quantum conditional entropy can be negative. In the following theorem, we find a lower bound for the conditional entropy and show that only entangled states can have negative conditional entropy.
\begin{theorem}
Let $\H$ be any conditional entropy, and let $A$ and $B$ be two Hilbert spaces. Then for any $\rho \in \mathfrak{D}\left(AB\right)$ the conditional entropy is lower bounded as,
\begin{equation}
 -\log\min\left(|A|,|B|\right) \leq  \H(A|B)_{\rho},
\end{equation}
The lower bound is achieved when $\rho_{AB}=\phi^{+}_{AB}$ is the maximally entangled state. Additionally, if $\rho_{AB}$ is a separable (non-entangled) state, then $ \H(A|B)_{\rho} \geq 0$.
\end{theorem}
\begin{proof}
Given that $\phi^{+}_{AB}$ can be used to teleport an arbitrary quantum state, it follows that there exists some channel $\mathcal{N} \in \operatorname{CUSC}\left(AB \rightarrow AB\right)$ such that $\rho_{AB} = \mathcal{N}\left(\phi^{+}_{AB}\right)$. It immediately follows that for any $\rho \in \mathfrak{D}\left(AB\right)$, then 
\begin{align}
   \H(A|B)_{\phi^{+}_{AB}} \leq \H(A|B)_{\rho_{AB}}
\end{align}
Thus, it is sufficient to demonstrate that $\H(A|B)_{\phi^{+}}$ is lower bounded by $-\text{log}\left(\min\{|A|, |B|\}\right)$. To this aim, we define $d \triangleq |A| = |B|$ and consider the state $\phi^{+}_{AB} \otimes \mathbf{u}_{A_{2}}$ where $|A_{2}| = d$. Then
\begin{align}
    & \H\left(AA_{2}|B\right)_{\phi^{+}_{AB} \otimes \mathbf{u}_{A_{2}}} \notag \\
    &=   \H\left(AA_{2}|B B_{2}\right)_{\phi^{+}_{AB} \otimes \mathbf{u} _{A_{2}B_{2}}} \\
    & = \H(A|B)_{\phi^{+}} + \H\left(A_{2} | B_{2}\right)_{\mathbf{u} _{A_{2} B_{2}}} \\
    & = \H(A|B)_{\phi^{+}} + \H\left(\mathbf{u}_{A_{2}}\right) \\
    &= \H(A|B)_{\phi^{+}} + \text{log}\left(d\right),
\end{align}
where the first line follows by introducing a trivial system $B_{2}$ such that $|B_{2}|=1$ and the second follows from conditional entropy.

We now demonstrate that $\H\left(AA_{2}|B\right)_{\phi^{+}_{AB} \otimes \mathbf{u}_{A_{2}}} \leq 0$ by demonstrating the existence of a channel $\mathcal{N} \in \operatorname{CUSC}\left(AA_{2}B \rightarrow AA_{2}\right)$ such that $\mathcal{N}\left(\phi^{+} \otimes \mathbf{u}\right) = \ket{\psi}\!\!\bra{\psi}_{AA_{2}}$ for some pure state $\ket{\psi} \in \mathcal{H}\left(AA_{2}\right)$. Consider the Choi matrix 
\begin{align}
    J^{\mathcal{N}}_{AA_{2} B \tilde{A} \tilde{A}_{2}} = \sum_{j=1}^{d^{2}} \phi_{AB}^{\left(j\right)} \otimes I _{A_{2}} \otimes \psi_{\tilde{A} \tilde{A}_{2}}^{\left(j\right)}
\end{align}
where $\{\ket{\phi}_{AB}^{\left(j\right)}\}|_{j=1}^{d^{2}}$ is an orthonormal basis of $AB$ consisting of maximally entangled states and where we set $\phi_{1} = \phi^{+}$. Likewise, $\{ \ket{\psi_{j}}_{\tilde{A} \tilde{A}_{2}} \}|_{j=1}^{d^{2}}$ is an orthonormal basis of $AA_{2}$. 

Evidently, $\mathcal{N}$ is trace-preserving and conditional unital (as $J^{\mathcal{N}}_{B \tilde{A} \tilde{A}_{2}} = I_{B \tilde{A} \tilde{A}_{2}} = J^{\mathcal{N}}_{B} \otimes \mathbf{u}_{\tilde{A} \tilde{A}_{2}}$). Additionally, $\mathcal{N}$ is trivially semi-causal as subsystem $B$ is traced out. From this, it follows that $\mathcal{N} \in \operatorname{CUSC}\left(AA_{2}B \rightarrow \tilde{A} \tilde{A}_{2}\right)$ and therefore 
\begin{align}
\H\left(AA_{2}|B\right)_{\phi^{1}_{AB} \otimes \mathbf{u}_{A_{2}}} 
 & \leq \H\left(\tilde{A} \tilde{A}_{2} | \tilde{B}\right)_{\mathcal{N}\left(\phi^{1}_{AB} \otimes \mathbf{u}_{A_{2}}\right)} \\
 &= \H\left( \mathcal{N}\left(\phi^{1}_{AB} \otimes \mathbf{u}_{A_{2}}\right) \right)
\end{align}
where $\tilde{B}$ is a trivial subsystem. Note that 
\begin{align}
    & \mathcal{N}\left(\phi^{+}_{AB} \otimes \mathbf{u}_{A_{2}}\right) \\
    &= \operatorname{Tr}_{AA_{2}B} \left[ J^{\mathcal{N}}_{AA_{2}B\tilde{A}\tilde{A}_{2}} \left( \phi^{+}_{AB} \otimes \mathbf{u}_{A_{2}} \otimes I_{\tilde{A} \tilde{A}_{2}} \right) \right] \\
    &= \operatorname{Tr}_{AA_{2}B} \!\Bigg[ \left(\sum_{j=1}^{d^{2}} \phi_{AB}^{\left(j\right)} \otimes I _{A_{2}} \otimes \psi_{\tilde{A} \tilde{A}_{2}}^{\left(j\right)} \right)  \phi^{+}_{AB} \notag \\
    & \qquad\qquad\qquad \otimes \mathbf{u}_{A_{2}} \otimes I_{\tilde{A} \tilde{A}_{2}}  \Bigg] \\
     &= \operatorname{Tr}_{AA_{2}B} \left[ \sum_{j=1}^{d^{2}} \phi_{AB}^{\left(j\right)} \phi^{+}_{AB} \otimes \mathbf{u}_{A_{2}} \otimes \psi_{\tilde{A} \tilde{A}_{2}}^{\left(j\right)} \right] \\
     &= \operatorname{Tr}_{AA_{2}B} \Big[\phi^{\left(1\right)}_{AB} \otimes \mathbf{u}_{A_{2}} \otimes \psi_{\tilde{A} \tilde{A}_{2}}^{\left(1\right)} \Big] \\
     &= \phi_{\tilde{A}\tilde{A}_{2}}^{\left(1\right)},
\end{align}
which is evidently a pure state. Hence, 
\begin{align}
    \H\left(AA_{2}|B\right)_{\phi^{1}_{AB} \otimes \mathbf{u}_{A_{2}}} & \leq  \H\left(\ket{\psi}\!\!\bra{\psi}_{\tilde{A}\tilde{A}_{2}}\right) \\
    &= 0.
\end{align}
Upon combining everything, we see that 
\begin{align}
    0 & \leq \H\left(A |B\right)_{\phi^{+}_{AB} \otimes \mathbf{u}_{2}} \\
    & \leq \H(A|B)_{\phi^{+}_{AB}} + \text{log}|A| \\
    & \leq 0 ,
\end{align}
from which it follows that $\H(A|B)_{\phi^{+}} = - \text{log}|A|$ and therefore $-\text{log}\left( \{|A|, |B|\} \right)$ lower bounds any density matrix $\rho_{AB}$. 

Finally, we demonstrate that if $\rho \in \mathfrak{D}\left(AB\right)$ is a separable state, then $\H(A|B)_{\rho} \geq 0$. This follows from noting that any separable state $\rho_{AB}$ may be written as 
\begin{align}
    \rho_{AB} = \sum_{j=1} p_{j} \psi_{A}^{\left(j\right)} \otimes \phi_{B}^{\left(j\right)},
\end{align}
where $\{\psi_{A}^{\left(j\right)} \}_j$ and $\{\phi_{B}^{\left(j\right)}\}_j$ are sets of pure states. It follows that the CUSC channel 
\begin{align}
   \mathcal{N} _{AB \rightarrow AB} \triangleq \sum_{j=1}^{n} p_{j} \mathcal{U}_{A \rightarrow A}^{\left(j\right)} \otimes \mathcal{V}_{B \rightarrow B}^{\left(j\right)}
\end{align}
satisfies 
\begin{align}
    \mathcal{N}\left(\ket{0}\!\!\bra{0}_{A} \otimes \ket{0}\!\!\bra{0}_{B}\right) &= \rho_{AB}
\end{align}
for a choice of unitaries $\{\mathcal{U}_{A}^{\left(j\right)}\}$ such that $\mathcal{U}^{j} \ket{0} = \ket{\phi_{A}^{\left(j\right)}}$ and a choice of unitaries $\{\mathcal{V}_{B}^{\left(j\right)} \}$ such that $\mathcal{V}^{\left(j\right)} \ket{0} = \ket{\psi_{B}^{\left(j\right)}}$. Then 
\begin{align}
    0 & = \H(A|B)_{\ket{0}\!\!\bra{0} _{A} \otimes \ket{0}\!\!\bra{0}_{B}} \\
    & \leq \H(A|B)_{\mathcal{N}\left(\ket{0}\!\!\bra{0} _{A} \otimes \ket{0}\!\!\bra{0}_{B}\right)} \\
      & \leq  \H(A|B)_{\rho_{AB}}.
\end{align}
This concludes the proof.
\end{proof}

Although conditional entropies are positive on all separable states, this does not mean that the conditional entropies are only positive on separable states. In fact, some entangled states have positive conditional entropy on \emph{all} choice of conditional entropy functions. The following lemma provides a simple criterion for finding states with positive conditional entropy.

\begin{lemma}
Let $\H$ be a conditional entropy, and let $\rho\in\mathfrak{D}\left(AB\right)$ with $d\eqdef|A|=|B|$ be a density matrix whose eigenvalues are less than or equal to $\frac{1}{d}$, and whose marginal state on system $B$ is a maximally mixed state, i.e., $\rho^B=\u_B$. Then, the conditional entropy on $\rho$ is non-negative, i.e.,
\begin{equation}
    \H(A|B)_{\rho}\geq 0\ .
\end{equation}
\end{lemma}
\begin{proof}

Suppose $\mathcal{N}_{\tilde{A}B\to AB'}$ is a $\tilde{A}\nrightarrow B$ semi-causal channel as shown in Fig.~\ref{fig_semi-casual_channel} and suppose $|B|=1$. For a semi-causal channel $\mathcal{N}_{\tilde{A}B\to AB'}$ such that 
\begin{equation}
    \mathcal{N}_{\tilde{A}B\to AB'}=\mathcal{E}_{R\tilde{A}\to A}\circ\mathcal{F}_{B\to RB'}\ ,
\end{equation}
$\mathcal{F}$ is required to be an isometry. When $|B|=1$, $\mathcal{F}$ is always an isometry. In this case, the channel $\mathcal{F}$ could be viewed as a bipartite state on system $RB'$. Let this state be the maximally entangled state $\phi_{RB'}^{+}$ and let $|R|=|B'|$. Then the resulting channel $\mathcal{N}_{\tilde{A}\to AB'}$ takes the form 
\begin{equation}
    \mathcal{N}_{\tilde{A}\to AB'}\left(\omega _{A}\right)=\mathcal{E}_{\tilde{A}\to AB'}\left(\omega _{A}\otimes \phi_{RB'}^{+}\right)\quad\forall\omega\in\mathfrak{D}\left(\tilde{A}\right)\ .
\end{equation}
Note that the channel $\mathcal{N}_{\tilde{A}\to AB'}$ is not necessarily conditional unital. It can be shown that channel $\mathcal{N}_{\tilde{A}\to AB'}$ is conditional unital iff the channel $\mathcal{E}_{\tilde{A}\to AB'}$  satisfies
\begin{equation}\label{eeee}
    \mathcal{E}_{\tilde{A}\to AB'}\left(\u _{A}\otimes \phi_{RB'}^{+}\right)=\u _{AB'}\ .
\end{equation}
Let $\mathcal{E}_{\tilde{A}\to AB'}$ be a quantum channel that satisfies \eqref{eeee}. Let $\sigma _{AB'}$ be the quantum state such that
\begin{equation}
    \sigma _{AB'}=\mathcal{N}_{\tilde{A}\to AB'}\left(|1\rangle\!\langle1|_{\tilde{A}}\right)\ ,
\end{equation}
where $|1\rangle$ is a pure state in $\tilde{A}$. We then have
\begin{align}
    \H\left(A|B'\right)_{\sigma}&=\H\left(A|B'\right)_{\mathcal{N}\left(|1\rangle\!\langle1|\right)}\\
    &\geq \H\left(|1\rangle\!\langle 1|_{\tilde{A}}\right)   \\
    &=0   .
\end{align}
To simplify the notation, we rename system $B'$ as $B$. Since $|R|=|B|=|\tilde{B}|$, we rename $R$ as $\tilde{B}$. Then, the channel $\mathcal{E}_{\tilde{A}R\to A}$ becomes $\mathcal{E}_{\tilde{A}B\to A}$, and its Choi matrix is denoted by $J^{\mathcal{E}}_{\tilde{A}\tilde{B}A}$. By \eqref{eeee}, we have
\begin{equation}\label{iiii}
    J^{\mathcal{E}}_{\tilde{A}\tilde{B}}=I _{A\tilde{B}}\quad \text{and}\quad   J^{\mathcal{E}}_{\tilde{B}A}=I_{\tilde{B}A}\ .
\end{equation}
Now, let $\rho_{AB}$ be as in the lemma, and define the channel $\mathcal{E}_{\tilde{A}B\to A}$ via its Choi matrix
\begin{equation}
    J^{\mathcal{E}}_{\tilde{A}\tilde{B}A} \eqdef|1\rangle\!\langle 1|_{\tilde{A}}\otimes \left(d\rho _{A\tilde{B}}\right)+\sum_{x=2}^d |x\rangle\!\langle x|_{\tilde{A}}\otimes\frac{I _{A\tilde{B}}-d\rho _{A\tilde{B}}}{d-1}\ .
\end{equation}
Note that the matrix above is positive semidefinite since $I _{A\tilde{B}}\geq d\rho _{A\tilde{B}}$ and it satisfies the two conditions in \eqref{iiii} since $\rho_{B}=\u_B$. With this choice of $\mathcal{E}$ we get 
\begin{align}
    &\sigma _{AB}=\mathcal{E}_{\tilde{A}R\to A}\left(|1\rangle\!\langle 1|_{\tilde{A}}\otimes\phi_{RB'}^{+}\right)\\
    &=\operatorname{Tr}_{\tilde{A}\tilde{B}}\left[\left(J^{\mathcal{E}}_{\tilde{A}\tilde{B}A} \otimes I^B\right)\left(|1\rangle\!\langle 1|_{\tilde{A}}\otimes\left(\phi_{\tilde{B}B}^{+}\right)^{T_{\tilde{B}}}\otimes I_A\right)\right]   \\
    &=\operatorname{Tr}_{\tilde{B}}\left[\left(\rho _{A\tilde{B}}\otimes I_{B}\right)\left(\Phi_{\tilde{B}B}\right)^{T_{\tilde{B}}}\otimes I _{A}\right]   \\
    &=\rho_{AB}   ,
\end{align}
where $T_{\tilde{B}}$ is the partial transposition map on system $\tilde{B}$. This completes the proof that the conditional entropy of $\sigma _{AB}$ is non-negative.
\end{proof}

\section{Examples of Conditional Entropies and Duality Relations}

We now demonstrate that any quantum relative entropy can be used to generate a corresponding quantum conditional entropy. 
\begin{theorem}
Let $\D$ be a quantum relative entropy. Then the function 
\begin{align*}
    \H^{\downarrow}\left(A \big| B\right)_{\rho} \triangleq \mathrm{log} \big| A \big| - \D\left(\rho_{AB} \big\| \mathbf{u}_{A} \otimes \rho_{B}\right)
\end{align*}
is a quantum conditional entropy.
\end{theorem}
\begin{proof}
First, we demonstrate that $\H$ satisfies the monotonicity property of conditional entropy. Let $\mathcal{N} \in \text{CUSC}(AB \rightarrow A'B')$ and consider the bipartite density matrix $\rho_{AB}$. We begin by considering the case where $A = A'$ and
\begin{align}
    \H^{\downarrow}(A \big| B)_{\mathcal{N}(\rho)} &= \text{log}|A| - \D\left(\mathcal{N}(\rho_{AB}) \big\| \mathbf{u}_{A} \otimes \text{Tr}_{A}\left[\mathcal{N}(\rho_{AB})\right]\right) 
\end{align}
and from using the fact that $\mathcal{N}$ is semi-causal and conditionally unital, we find that
\begin{align}
    \mathbf{u}_{A} \otimes \text{Tr}_{A}\left[\mathcal{N}\left(\rho_{AB}\right)\right] &= \mathbf{u}_{A} \otimes \text{Tr}_{A}\left[\mathcal{N}\left(\mathbf{u}_{A} \otimes \rho_{B}\right)\right] \\
    &= \mathcal{N}(\mathbf{u}_{A} \otimes \rho_{B})
\end{align}
from which we have 
\begin{align}
    \H^{\downarrow}(A \big| B)_{\mathcal{N}(\rho)} &= \text{log}|A| - \D\left(\mathcal{N}(\rho_{AB}) \big\| \mathcal{N}(\mathbf{u}_{A} \otimes \rho_{B}) \right) \\
    &\geq \text{log}|A| - \D\left(\rho_{AB} \big\| \mathbf{u}_{A} \otimes \rho_{B} \right) 
\end{align}
where the last line follows from the data processing inequality.

We now prove that any entropy $\H$ satisfying the form given in the theorem statement is invariant under isometries. Consider any isometry channel $\mathcal{V} \in \text{CPTP}(A \rightarrow A')$ and let $C$ represent a Hilbert space where $|C| = |A'| - |A|$. Then 
\begin{align}
    \H^{\downarrow}\left(A \big| B \right)_{\mathcal{V}(\rho)} = \text{log}|A'| - \D(\mathcal{V}_{A \rightarrow A'}(\rho_{AB}) \big\| \mathbf{u}_{A'} \otimes \rho_{B})
\end{align}
Clearly, if $\mathcal{V}$ is unitary such that $A \simeq  A'$, then $\H(A' \big| B)_{\mathcal{V}(\rho)} = \H(A \big| B)_{\rho}$. Then w.l.o.g. we assume that 
\begin{align}
    \mathcal{V}_{A \rightarrow A'}(\rho_{AB}) &= \rho_{AB} \oplus \mathbf{0}_{CB} \\
    &= \begin{pmatrix}
    \rho_{AB} & \mathbf{0} \\
    \mathbf{0} & \mathbf{0}_{CB}
    \end{pmatrix}
\end{align}
since the conditional entropy is invariant under unitaries. We can additionally rewrite $\mathbf{u}_{A'} = t \mathbf{u}_{A} \oplus (1-t) \mathbf{u}_{C}$ where $t \triangleq \frac{|A|}{|A'|}$. From this, it follows that 
\begin{align}
    & \D(\mathcal{V}_{A \rightarrow A'}(\rho_{AB})\big\| \mathbf{u}_{A'} \otimes \rho_{B}) \\
    &= \D\left(\rho_{AB} \oplus \mathbf{0}_{CB} \big\| \left(t \mathbf{u}_{A} \otimes \rho_{B}\right)\oplus \left(\left(1-t\right)\mathbf{u}_{C} \otimes \rho_{B} \right)\right) \\
    &= \D\left(\rho_{AB} \big\| \left(t \mathbf{u}_{A} \otimes \rho_{B}\right)\right) \\
    &= \D\left(\rho_{AB}  \big\| \left(\mathbf{u}_{A} \otimes \rho_{B} \right)\right) - \text{log}(t)
\end{align}
such that upon substitution, we have 
\begin{align}
    \H^{\downarrow}\left(A \big| B\right)_{\mathcal{V}(\rho)} &= \text{log}|A'| - \D\left(\rho_{AB} \big\| \mathbf{u}_{A} \otimes \rho_{B} \right) + \text{log}(t) \\
    &= \H\left(A \big| B\right)_{\rho}
\end{align}

The additive property of $\H$ follows immediately from the additive property of $D$. Finally,  when $\mathbf{u} = \frac{\mathbb{I}}{2}$, then
\begin{align}
    \H\left(A \big| B\right)_{\mathbf{u}} = \text{log}(2) - \D(\mathbf{u} \Vert \mathbf{u}) = 1
\end{align}
and so normalisation is also satisfied. 
\end{proof}
It follows that the von-Neumann conditional entropy is a conditional entropy under our axiomatic approach. The von-Neumann conditional entropy is defined as 
\begin{align}
    H\left(A \big| B\right)_{\rho} & \triangleq \text{log}|A| - D(\rho_{AB} \big| \mathbf{u}_{A} \otimes \rho_{B})
\end{align}
where $D$ is the Umegaki relative entropy such that for any $\rho_{AB}$ and $\sigma_{AB}$, then
\begin{align}
    D(\rho_{AB} \big| \sigma_{AB}) = \text{Tr}\left[\rho_{AB} \left( \text{log}\left(\rho_{AB}\right)- \text{log}\left(\sigma_{AB}\right)\right)\right]
\end{align}

Finally, we define the dual of a conditional entropy function and demonstrate that the conditional von Neumann entropy is self-dual.

\begin{definition}
Let $\H$ be a conditional entropy. For any density matrix $\rho \in \mathfrak{D}(AB)$ with a purification $\varphi \in \mathfrak{D}(ABC)$ where $\rho_{AB} = \text{Tr}_{C}(\varphi_{ABC})$ the dual of $\H$ is defined as 
\begin{align}
    \H^{\mathrm{dual}}\left( A \ \big| \ B \right)_{\rho} \triangleq -\H\left( A \ \big| \ C \right)_{\rho}
\end{align}
\end{definition}

Let $H$ be the conditional von-Neumann entropy. Then it follows from the relation $H\left( A \big| B \right)_{\varphi} + H\left( A \big| C\right)_{\varphi} = 0$ that 
\begin{align}
    H^{\text{dual}}\left(A \big| B\right)_{\rho} & = -H(A \big| C)_{\varphi} \\
    &= H\left(A \big| B\right)
\end{align}
and therefore the conditional von-Neumann entropy is self-dual.

\section{Quantum conditional majorization reduces to classical conditional majorization}
In~\cite{Gour}, a pre-order $\succsim_{C}$ between two joint probability distributions (i.e. classical states) denotes \emph{classical conditional majorization}. In this section, we demonstrate that the axiomatic approach to quantum conditional majorization defined in the previous section is a generalization of classical conditional majorization. To this aim, we demonstrate that when we restrict the underlying bipartite quantum states to classical states, the quantum conditional majorization is equivalent to the classical conditional majorization. 

Suppose the system $A\eqdef X$ is classical, we call a channel $\mathcal{M}\in$CPTP$(XB\to XB')$ \emph{conditionally doubly stochastic} (CDS in short) if it has the following form 
\begin{equation}
    \mathcal{M}_{YB\to YB'}=\sum_{j}\mathcal{D}_{X\to X}^{(j)}\otimes\mathcal{F}_{B\to B'}^{(j)}
\end{equation}
where each $\mathcal{D}{(j)}$ is a classical doubly stochastic channel and each $\mathcal{F}_{B\to B'}^{(j)}$ is a completely positive (CP) map such that $\sum_j \mathcal{F}_{B\to B'}^{(j)}$ is CPTP.

\begin{lemma}\label{lem_red}
Suppose all the systems involved are classical with $A\eqdef X$, $B\eqdef Y$, and $B'\eqdef Y'$. Then we have
\begin{equation}
    \operatorname{CDS}(XY\to XY')\subseteq \operatorname{CUSC}(XY\to XY')
\end{equation}
\end{lemma}
\begin{proof}
Recall that if $\mathcal{M} \in \text{CDS}(XY \rightarrow XY')$, there must exist a set of classical doubly stochastic channels $\{\mathcal{D}^{(j)}\}$ and a quantum instrument $\{\mathcal{F}_{Y \rightarrow Y'}^{(j)}\}$ s.t.  
\begin{equation}
    \mathcal{M}_{AB\to AB'}=\sum_{j}\mathcal{D}_{A\to A}^{(j)}\otimes\mathcal{F}_{B\to B'}^{(j)}
\end{equation}
Then for any input of the form $\mathbf{u}_{A} \otimes \sigma_{B}$, it follows that 
\begin{align}
    \mathcal{M}_{AB\to AB'}(\mathbf{u}_{A} \otimes \sigma_{B})&=\sum_{j}\mathcal{D}_{A\to A}^{(j)}(\mathbf{u}_{A}) \otimes\mathcal{F}_{B\to B'}^{(j)}(\sigma_{B}) \\
    &= \mathbf{u}_{A} \otimes \sum_{j}\mathcal{F}_{B\to B'}^{(j)}(\sigma_{B})
\end{align}
and therefore $\mathcal{M}$ is conditionally unital. (Note that the third line of the above follows from the fact that $\mathcal{D}^{(j)}$ is doubly stochastic, so $\mathcal{D}^{(j)}(\mathbf{u}_{A}) = \mathbf{u}_{A}$ for all $j$.)

Finally, we show that $\mathcal{M}$ is semi-causal as follows:
\begin{align}
   \text{Tr}_{A}&\left[ \mathcal{M} \circ \mathcal{T}_{A \rightarrow A} \right] = \text{Tr}_{A}\left[ \sum_{j}\mathcal{D}_{A\to A}^{(j)} \circ \mathcal{T}_{A \rightarrow A} \otimes\mathcal{F}_{B\to B'}^{(j)} \right] \\
    &=   \sum_{j} \text{Tr}_{A}\left[\mathcal{D}_{A\to A}^{(j)} \circ \mathcal{T}_{A \rightarrow A} \right] \otimes\mathcal{F}_{B\to B'}^{(j)}  \\
     &=   \sum_{j} \text{Tr}_{A}\left[\mathcal{D}_{A\to A}^{(j)} \right] \otimes\mathcal{F}_{B\to B'}^{(j)}  \\
     &= \text{Tr}_{A}\left[ \mathcal{M}_{AB \rightarrow AB'} \right]
\end{align}
Given that $\mathcal{M}$ is conditionally unital and semi-causal, then the theorem statement follows.
\end{proof}

\begin{theorem}
Let $\rho_{XY}$ and $\sigma_{XY'}$, be classical bipartite states such that $A=X$, $B=Y$, and $B'=Y'$ are all classical systems. Then $\sigma_{XY'} \precsim \rho_{XY}$ if and only if $\sigma_{XY'} \precsim_{C} \rho_{XY}$ where $\precsim_{C}$ is the classical conditional majorization defined in~\cite{Gour}. 
\end{theorem}
\begin{proof}

It then immediately follows from lemma~\ref{lem_red} that $\rho_{XY}\succsim \sigma_{XY'}$ as in definition 4. 

To show the other direction, we suppose that $\rho_{XY}\succsim \sigma_{XY'}$ as in definition 4. That is, there exists $\mathcal{N}\in\operatorname{CUSC}(XY\to XY')$ such that
\begin{equation}
    \sigma_{XY'}=\mathcal{N}_{XY\to XY'}(\rho_{XY})\ .
\end{equation}
Since $\mathcal{N}$ is semi-causal, it has the form 
\begin{equation}
    \mathcal{N}_{XY\to XY'}=\sum_{j=1}^{k}\mathcal{E}_{X\to X}^{(j)}\otimes\mathcal{F}_{Y\to Y'}^{(j)}\ ,
\end{equation}
where for each $j\in [K]$, $\mathcal{E}^{(j)}\in\text{CPTP}(X\to X)$ and $\mathcal{F}^{(j)}\in \text{CP}(Y\to Y')$ with $\sum_j\mathcal{F}_j\in \text{CPTP}(Y\to Y')$. Hence,
\begin{equation}
  \sigma_{XY'}=\sum_{j=1}^k \mathcal{E}_{X\to X}^{(j)}\otimes\mathcal{F}_{Y\to Y'}^{(j)}(\rho_{XY})\ . 
\end{equation}
Since $\rho,\sigma$ are classical, we have the identification $\rho_{XY}\cong P$ and $\sigma_{XY'}\cong Q$. Thus, the above equation can be expressed as
\begin{equation}
    Q=\sum_{j=1}^k E^{(j)}PR^{(j)}\ ,
\end{equation}
where for each $j\in[k]$, $E^{(j)}$ is the transition matrix of $\mathcal{E}^{(j)}$, and where $R^{(j)}$ is the transpose of transition matrix of $\mathcal{F}^{(j)}$. Hence, each $E^{(j)}$ is column stochastic and $R\eqdef\sum_{j=1}^kR^{(j)}$ is row stochastic. Let $\{\p_{y}\}_{y\in[n]}$ be the columns of $P$, and observe that in components form, the above relation can be expressed as 
\begin{equation}\label{8.222}
    q_{xw}=\sum_{j=1}^k\sum_{y=1}^n (E^{(j)}\p_y)_x r_{yw}^{(j)}\ .
\end{equation}
So far we only used the constraint that $\mathcal{N}$ is semi-causal. The other constraint that $\mathcal{N}$ is conditionally unital implies that if $p_{xy}=\frac{1}{m}p_y$ then $q_{yw}=\frac{1}{m}q_{w}$. Therefore, for simplicity we take $p_{xy}=\frac{1}{m}\delta_{y}{y_0}$ for some fixed $y_0\in[n]$. Then the above equation becomes 
\begin{equation}\label{8.223}
    \frac{1}{m}q_{w}=\sum_{j}(E^{(j)}\u)_x r^{(j)}_{y_0 w}\ .
\end{equation}
Summing over $x$ on both sides of the equation, and using the fact that each $E^{(j)}$ is column stochastic implying that $E^{(j)}\u$ is a probability vector, we have
\begin{equation}
    q_w=\sum_j r_{y_0 w}^{(j)}\ .
\end{equation}
Note that we assumed that $y_0$ was fixed, and therefore as a result the equation above does not imply that $q_w$ is independent of $y_0$ if we allow $y_0$ to vary. In other words, each value of $y_0$ corresponds to a different $P$ matrix and therefore possibly resulting in a different $Q$ matrix. For simplicity, we rename $y_0$ as $y$ and denote 
\begin{equation}
    q_{w}=t_{yw}\eqdef\sum_{j}r^{(j)}_{yw}\ .
\end{equation}
By definition, the matrix $T=(t_{yw})$ is row-stochastic, and with the above notations , we can express \eqref{8.223} as
\begin{equation}
    \u=\sum_{j}\frac{r_{yw}^{(j)}}{t_{yw}}E^{(j)}\u\quad y\in[n],\ w\in[\ell]\ .
\end{equation}
Therefore, the matrices
\begin{equation}
    D_{(y,w)}\eqdef\sum_j \frac{r_{yw}^{(j)}}{t_{yw}}E^{(j)}\quad y\in[n],\ w\in[\ell]
\end{equation}
are all doubly stochastic. To see this, observe that for each $y$ and $w$, the set $\{\frac{r_{yw}^{(j)}}{t_{yw}}\}$ forms a probability distribution over $j$, and the fact that each $E^{(j)}$ is column stochastic implying that $D_{(y,w)}$ is column stochastic. Moreover, the relation $\u=D_{(y,w)}u$ implies that $D_{(y,w)}$ are row stochastic. Finally, for general probability matrices $P=[\p_1,\ldots,\p_n]$ and $Q=[\q_1,\ldots,\q_{n'}]$, we conclude that \eqref{8.222} is equivalent to 
\begin{equation}
    \q_w=\sum_{y=1}^nt_{yw}D_{(y,w)}\p_y\quad\forall w\in[n']\ .
\end{equation}
The above relation is precisely the condition given in lemma 3 of~\cite{Gour}. Therefore, we conclude that $\rho_{XY}\succsim \sigma_{XY'}$ implies classical conditional majorization. 
\end{proof}

\section{Connection to Games of Chance}

\subsection{Games of Chance for Classical Bipartite States}

In~\cite{brandsen2021entropy}, games of chance were utilized to find an ordering between classical bipartite states, and it has been demonstrated that such games of chance are an alternate, operational approach for characterising classical majorisation. We aim to likewise construct a set of quantum games of chance which provide an operational meaning for majorisation of quantum bipartite states.

Firstly, we briefly review the games of chance introduced in~\cite{brandsen2021entropy} for classical bipartite states.  We first note that any classical bipartite state may be represented as a probability distribution of the form $\mathbf{p} = \{p_{yz}\}$ where the corresponding state in quantum notation is $\rho_{AB} = \sum_{y,z} p_{y,z} \ket{y}\!\!\bra{y} \otimes \ket{z}\!\!\bra{z}$. Each game is fixed by a classical channel $\mathcal{T}$ with transition probabilities $\{t_{w|z'}\}$, as well as the bipartite state $\mathbf{p}$ which is being used for the game. 

The player is only allowed to access subsystem $B$ (corresponding to the value of $z$), and the goal is for the player to correctly guess the value of $y$. The family of gambling games includes all possible gambling games that incorporate a correlated source. Given access to $z$, the player chooses $z' = f(z)$ and the host subsequently selects $w$ from the conditional distribution $\mathcal{T}$. In general, the player will choose $f(z)$ based on their knowledge of $z$, as well as the fixed distributions $\{p_{yz}\}$ and $\{t_{w|z'}\}$. Finally, once $w$ is selected, the player communicates a set $\mathcal{S}$ of $w$ guesses to the host and wins if and only if $y \in \mathcal{S}$. The game is depicted in Fig.~\ref{classical1}.  

\begin{figure}[h]
\begin{overpic}[scale=.56]{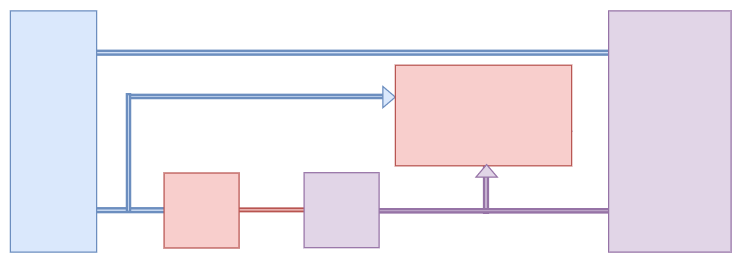}
\put(112,17){$\mathcal{T}$}
\put(105,80){$y$}
\put(65,17){$f$}
\put(90,23){$z'$}
\put(139.5,61){Choose set}
\put(210,57){Win if}
\put(210,45){$y \in \mathcal{S}$}
\put(139.5,51){$\mathcal{S}$ with $w$}
\put(139.5,40){elements}
\put(35,25){$z$}
\put(142,23){$w$}
\put(7, 50){$\{p_{yz}\}$}
\end{overpic}
 \caption{\linespread{1}\selectfont{\small A classical gambling game with a correlated source. The player is provided with the value $z$. Based on this value, the player chooses $z$ (or the function $f$) and sends it to the host. The host then chooses the $w$ game based on a (possibly incomplete) distribution matrix $T=(t_{w|z})$. The player will then form a set $\mathcal{S}$ containing $w$ guesses based on $w$ and $z$, and will win the game if $y \in \mathcal{S}$.}} 
  \label{classical1}
\end{figure}
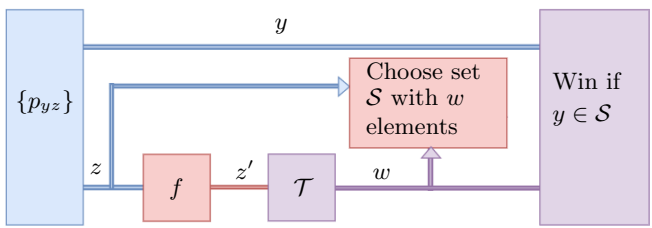

It has previously been demonstrated~\cite{brandsen2021entropy} that the conditional majorisation relation $\precsim_{C}$ is equivalent to the partial ordering induced by classical games of chance, namely $\sigma_{AB} \precsim \rho_{AB'}$ if and only if  $\text{Prob}_{\mathcal{T}}(\mathbf{q}) \leq \text{Prob}_{\mathcal{T}}(\mathbf{p})$
for all conditional distributions $\mathcal{T}$ where $\text{Prob}_{\mathcal{T}}(\mathbf{p})$ is the expected success probability for playing the classical gambling game characterised by $\mathcal{T}$ with state $\mathbf{p}$.

\subsection{Extension of gambling games to quantum bipartite states}

In this section, we build on the framework for classical games of chance to develop the most general games of chance for quantum bipartite states (and respectively quantum channels). The only restriction on such games is that the ordering induced by these families of games of chance must reasonably correspond to a measure of uncertainty. For example, states which are equivalent up to a unitary must have equivalent performance in these games of chance. 

First, we provide a straightforward extension of the classical gambling games, and demonstrate that this direct extension is insufficient to capture necessary quantum effects such quantum entanglement being a resource. Next, we provide a broader set of games which still reduces to the same ordering on classical states, but which leads to negative entropy for maximally entangled states. 

To generalize the classical gambling game, we first transform Bob's observation of the classical output $z$ into a POVM implemented by Bob on system $B$. Evidently, $y$ can be written in quantum notation as $\ket{y}\!\!\bra{y}$ such that a measurement in the computational basis allows perfect observation of $y$.

Likewise, Bob's choice of $w$ guesses is equivalent to sending an ordered list of all potential outcomes and winning if $y$ is in the first $w$ elements of the list. In turn, this can be transformed into performing an ordered measurement on system $A$ with the output being transmitted to the host. For example, if their first guess is $2$, the second guess is $1$, and the third guess is $3$, then the they would implement the ordered measurement $\left\{ \ket{2}\!\!\bra{2}, \ket{1}\!\!\bra{1}, \ket{3}\!\!\bra{3} \right\}$ on Alice's system. This game is depicted in Fig.~\ref{firstgame}. 
\begin{figure}[h!]
  \begin{overpic}[scale=0.555]{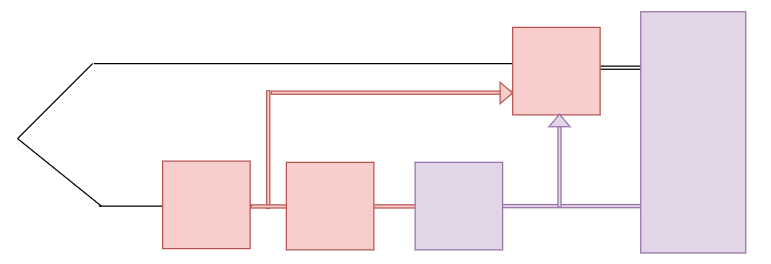}
  \put(-2, 51){$\rho_{AB}$}
 \put(27, 69){$A$}
  \put(27, 10){$B$}
\put(179, 63){$\hat{\Pi}^{*}_{z, w}$}
\put(220, 50){Win if}
\put(220, 35){$y \leq w$}
\put(149, 18){$\mathcal{T}$}
\put(108, 18){$f$}
\put(65, 18){$\hat{\Pi}$}
\put(89, 12){$z$}
\put(82, 45){$z$}
\put(129, 23){$z'$}
\put(175, 23){$w$}
\end{overpic}
\caption{Direct extension of classical gambling games for bipartite states. The player selects any rank-one measurement $\hat{\Pi}$ to implement on system $B$ and observes the outcome $z$. Given $z$, the player selects input $z'=f(z)$ to $\mathcal{T}$, and the hosts draws $w$ according to $\{t_{w|z'}\}$. The player then implements any rank one measurement $\hat{\Pi}^{*}_{z,w}$ given knowledge of $z$ and $w$, and wins if the output $y\leq w$.}
\label{firstgame}
\end{figure}
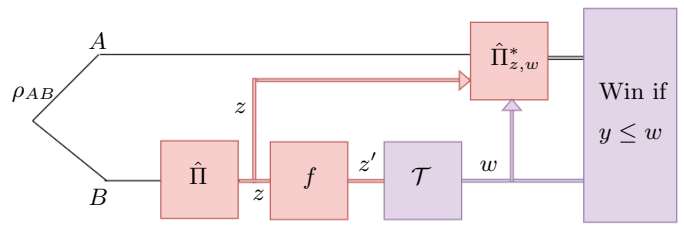
However, this direct extension is insufficient to distinguish between the un-entangled product state $\sigma_{AB} = \ket{0}\!\!\bra{0}_{A} \otimes \ket{0}\!\!\bra{0}_{B}$ and the entangled product state $\rho_{AB} = \phi_{AB}^{+}$. More specifically, any game can be won with certainty when played with $\ket{0}\!\!\bra{0}_{A} \otimes \ket{0}\!\!\bra{0}_{B}$. Regardless of the player's choice of $\hat{\Pi}$ or $f$, and regardless of $w$, the player can simply implement the measurement $\hat{\Pi}^{*}_{z,w} \triangleq \big\{\ket{0}\!\!\bra{0}, \ket{1}\!\!\bra{1} \big\}$ on the first system. The first outcome (corresponding to measurement element $\ket{0}\!\!\bra{0}$) will be obtained with certainty, and so $y=1$ will be less than or equal to any $w \in \{1,2,3,4\}$ leading to a win for the player. 

To break this equivalency, we now introduce a more sophisticated game which allows an adversarial player to ``scramble'' the state of system $A$ with some probability $p$.  We define a conditional quantum gambling game $G=(\mathcal{T},\mathbf{p})$ by the classical channel $\mathcal{T}$ with transition matrix is $T = (t_{w|z'} )_{w,z'}$, and the probability $p$ of adversarial involvement. The specific steps of the game are as follows: 
\begin{enumerate}
    \item Alice and Bob are given a classical description of the state they share as the entries of the corresponding density matrix $\rho_{AB}$.
    \item The host generates $x \in \{0, 1\}$ from the probability distribution $\mathbf{p} = \left(p, 1-p\right)$, where $p\in[0,1]$. If $x = 0$, then the adversary is allowed to interfere by implementing a measurement $\hat{\Pi}^{*}$ on system $A$. 
    \item Bob is told which measurement the adversary implemented (if any), but is not told the result. He is then allowed to perform any measurement $\hat{\Pi}' = \{\Pi^{' (z)}_{B}\}$ on system $B$ and is told the measurement result $z$.
    \item Bob can choose any $z' = f\left(z\right)$ to input to the classical channel $\mathcal{T}$ with corresponding transition matrix $T = (t_{w|z'} )_{w,z'}$. The host draws $w$ from $\{t_{w|z'}\}$
    \item Given $z$ and $w$, Alice implements a final rank-one projective measurement $\hat{\Pi}(z,w) = \{\Pi^{(k)}(z,w)\}|_{k}$ on system $A$. The rank-one requirement follows from noting that if Alice were allowed to choose a trivial measurement such as $\{I, 0, \ldots, 0\}$, she would win any game with certainty. 
    \item Alice and Bob win if the outcome $y$ is less than $w$.
\end{enumerate}

\begin{figure}[h!]
  \begin{overpic}[scale=0.5]{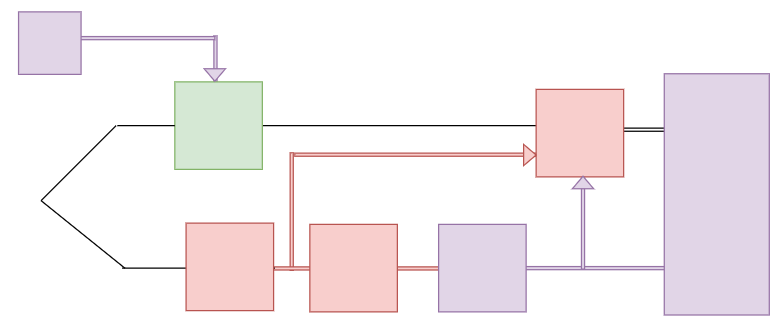}
 \put(-7, 40){$\rho_{AB}$}
 \put(20, 55){$A$}
  \put(20, 17){$B$}
  \scriptsize
\put(162, 57){$\hat{\Pi}(z, w)$}
\normalsize
\put(202, 50){Win if}
\put(202, 38){$y \leq w$}
\put(140, 17){$\mathcal{T}$}
\put(103, 15){$f$}
\put(64, 17){$\hat{\Pi}'$}
\put(85, 11){$z$}
\put(80, 40){$z$}
\put(122, 22){$z'$}
\put(185, 22){$w$}
\put(12, 85){$\mathbf{p}$}
\put(61, 58.5){$\hat{\Pi}^{*}$}
\normalsize
\put(45, 89){$x$}
\end{overpic}
\caption{Alternative extension of classical gambling games for bipartite states. With probability $p$, $x=0$ and the adversarial player is allowed to implement an adversarial measurement $\hat{\Pi}_{A}^{*}$ on Alice's system at the beginning of the game. The player selects any rank-one measurement $\hat{\Pi}$ to implement on system $B$ and observes the outcome $z$. Given $z$, the player selects input $z'=f(z)$ to $\mathcal{T}$, and the hosts draws $w$ according to $\{t_{w|z'}\}$. The player then implements any rank one measurement $\hat{\Pi}(z,w)$ given knowledge of $z$ and $w$, and wins if the output $y\leq w$.}
\label{firstgame}
\end{figure}
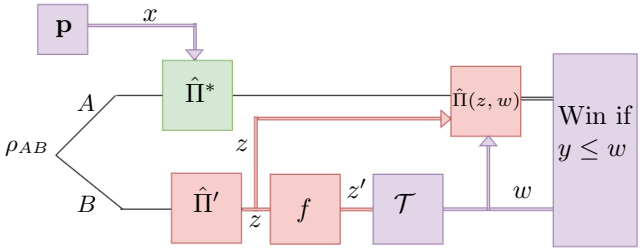

In Appendix~\ref{appendixC}, we demonstrate that the ordering induced by this extended quantum gambling game on classical bipartite states reduces to classical conditional majorisation. In the following, we will describe the reward function for the above games of chance and demonstrate that the induced partial ordering between quantum bipartite states distinguishes between entangled and unentangled pure states. Additionally, in the case where system $B$ is trivial, the ordering reduces to standard majorisation.

First, we denote by $\mathcal{R}_{\mathcal{T}, p}$ the total expected reward for playing the game corresponding to $\mathcal{T}$ and probability distribution $\mathbf{p} = \{p, 1-p\}$. Evidently, since $X$ is drawn before the players are allowed to make any moves, then 
\begin{align}
    \mathcal{R}_{\mathcal{T}, p} = p \times \mathcal{R}_{\mathcal{T}, 0} + (1-p) \times \mathcal{R}_{\mathcal{T}, 1}
\end{align}
In other words, the total expected reward is a convex combination of the total expected reward when the adversary is allowed a measurement and the total expected reward when the adversary is not allowed a measurement.

Suppose that Alice and Bob have a fixed strategy- namely, that Bob implements measurement $\hat{\Pi}^{'}$, selects $z' = f(z)$, and Alice implements the optimal measurement $\hat{\Pi}(z,w)$. We first consider the case where $p = 0$, i.e. there is no adversary. Then no action is taken prior to Bob's first measurement $\hat{\Pi}'$ on subsystem $B$, and the probability of observing a measurement result corresponding to  $\Pi_{\text{B}}^{'}(z)$ is
\begin{align}
 P\left(z \big| \hat{\Pi}_{\text{B}}^{'}\right)
 & = \text{Tr} \left[ \rho_{AB} \Pi_{\text{B}}^{\prime (z)} \right] .
\end{align}
The corresponding post-measurement state is
\begin{align}
    \rho_{z} \left(\hat{\Pi}_{\text{B}}', \rho_{AB} \right) \triangleq   \frac{ \Pi_{\text{B}}^{\prime (z)}  \rho_{AB} \Pi_{\text{B}}^{'(z)}}{ P\left(z \big| \hat{\Pi}_{\text{B}}^{'}\right)}.
\end{align}
In \emph{any} game, Alice will always choose $\hat{\Pi}_{A}(z,w)$ to maximize the reward. Namely, she will measure in the eigenbasis of $\rho_{z}(\hat{\Pi}_{B}^{'}, \rho_{AB})$ with the first measurement element corresponding to the eigenvector of $\rho_{z}(\hat{\Pi}_{B}^{'}, \rho_{AB})$ with the largest eigenvalue and so on. Then the success probability  for this fixed strategy is
\begin{align}
 & \sum_{z} P(z \big| \hat{\Pi}') \sum_{k=1}^{w} t_{w|f(z)} \left\|\Pi_{\text{A}}^{(k)}(z,w) \text{Tr}_{B}\left[ \rho_{z}\!\left(\hat{\Pi}_{\text{B}}', \rho_{AB} \right) \right]\right\| \\
  &= \sum_{z} P(z \big| \hat{\Pi}') \sum_{w} t_{w|f(z)} \left\| \text{Tr}_{B}\left[ \frac{\Pi_{\text{B}}^{\prime (z)}  \rho_{AB} \Pi_{\text{B}}^{'(z)}}{P\left(z \Big| \hat{\Pi}_{B}'\right)} \right]\right\|_{(w)} \\
    &= \sum_{z} \sum_{w} t_{w|f(z)} \left\| \text{Tr}_{B}\left[\Pi_{\text{B}}^{\prime (z)}  \rho_{AB} \Pi_{\text{B}}^{'(z)}\right]\right\|_{(w)} \\
\end{align}
 Naturally, Bob will wish to choose the remaining game components (that is, $\hat{\Pi}_{B}^{'}$ and $f$) to maximize their reward, s.t. the final reward function for the game $G = \{\mathcal{T}, p = 0\}$ may then be expressed as
\begin{align}
    R_{G}\left(\rho_{AB}\right) &\triangleq 
     \underset{\hat{\Pi}^{'},  f }{\text{max}} \left( \sum_{z, w} t_{w|f(z)} \left\|  \text{Tr}_{B}\left[\Pi_{\text{B}}^{\prime (z)}  \rho_{AB} \Pi_{\text{B}}^{'(z)}\right]\right\|_{(w)} \right) 
\end{align}
Now suppose the adversarial player is allowed to interfere with probability p. If the adversary chooses measurement $\hat{\Pi}_{A}^{*} = \{\Pi_{A}^{*(j)}\}|_{j}$, then the starting state $\rho_{AB}$ will be transformed to $\sum_{j} \Pi_{A}^{*(j)} \rho_{AB} \Pi_{A}^{*(j)}$. Thus, following an adversarial measurement, Alice and Bob will effectively be playing the quantum gambling game with the transformed state $\sum_{j} \Pi_{A}^{*(j)} \rho_{AB} \Pi_{A}^{*(j)}$ and will have an expected reward of $R_{\mathcal{T}, 0}\left(\sum_{j} \Pi_{A}^{*(j)} \rho_{AB} \Pi_{A}^{*(j)} \right)$. 

The adversary will wish to choose the ``worst case'' $\hat{\Pi}_{A}^{*}$ to minimize Alice and Bob's reward, leading to the following reward function for a general quantum bipartite gambling game

\begin{align}
    R_{\mathcal{T}, p}\left(\rho_{AB}\right) &\triangleq p \times  \underset{\hat{\Pi}_{A}^{*}}{\text{min}}\left(  R_{\mathcal{T}, 0}\left(\sum_{j} \Pi_{A}^{*(j)} \rho_{AB} \Pi_{A}^{*(j)} \right)  \right) \\
    & + (1-p) R_{\mathcal{T}, 0}\left(\rho_{AB}\right) 
\end{align}

Finally, we are ready to introduce conditional majorisation for quantum bipartite states based on games of chance.
\begin{definition}
The state $\rho_{AB}$ majorizes $\sigma_{AB'}$ with respect to subsystem $A$, based on quantum conditional gambling games, denoted as $\sigma_{AB'} \precsim_{\left(A\right)} \rho_{AB}$, if
\begin{align}
    R_{\mathcal{T}, p}\left(\sigma_{AB'}\right) \leq R_{\mathcal{T}, p}\left(\rho_{AB}, U\right)
\end{align}
for all $0 \leq p \leq 1$, $\mathcal{T} = \{t_{w'|z'}\}$, $w' \in \{1, \ldots , |A| \}$, and $z' \in \{1, \ldots , \max\left(|B|, |B'|\right)\}$.
\end{definition}

It follows immediately from the definition of the reward function that it is sufficient to restrict to the set of games for which $p \in \{0, 1\}$, namely, games in which the adversary either always interferes or never interferes. In the remainder of this work we restrict to this set unless otherwise specified.

One key property of this partial ordering is that maximally entangled pure states outperform tensor product pure states. 
\begin{lemma}
Consider the case in which $d= |A| = |A'| = |B|$ and $d \geq 2$. Then any tensor product state $\sigma_{AB'} = \sigma_{A} \otimes \sigma_{B'}$ is majorized by $\phi_{AB}^{(+)}$, namely,
\begin{align}
    \sigma_{AB'} \prec_{\left(A\right)} \phi_{AB}^{(+)}.
\end{align}
\end{lemma}

\begin{proof}
First, we demonstrate that $R_{\mathcal{T}, p} \left(\phi_{AB}^{(+)}\right) = 1$ for all $\mathcal{T}$ and all $p$. To do this, it is sufficient to show that 
\begin{align}
 R_{\mathcal{T}, 0} \left( \sum_{j} \Pi_{A}^{*(j)} \rho_{AB} \Pi_{A}^{*(j)} \right)  =1
\end{align}
for all $\mathcal{T}$ and all potential adversarial projective measurements $\Pi_{A}^{*}$. Since the adversarial measurement is projective, there exists some basis $\{ \ket{\psi_{j}}\}|_{j=1}^{n}$ and some partition of that basis $\{\mathcal{S}_{k}\}$ s.t. 
\begin{align}
    \Pi_{A}^{*(k)} = \sum_{j \in \mathcal{S}_{k}} \ket{\psi_{j}}\!\!\bra{\psi_{j}}
\end{align}
Note that 
\begin{align}
    & \sum_{j} \Pi_{A}^{*(j)} \phi_{AB}^{(+)} \Pi_{A}^{*(j)} = \frac{1}{d} \sum_{j, k, \ell} \Pi_{A}^{*(j)} \ket{k}\!\!\bra{\ell}_{A} \Pi_{A}^{*(j)} \otimes \ket{k}\!\!\bra{\ell}_{B}   \\
    & = \frac{1}{d} \sum_{z_{1}, z_{2} \in \mathcal{S}_{j}} \sum_{k, \ell} \ket{\psi_{z_{1}}}\!\!\braket{\psi_{z_{1}}| k}\!\!\braket{\ell|\psi_{z_{2}}}\!\!\bra{\psi_{z_{2}}} \otimes \ket{k}\!\!\bra{\ell}_{B}   \\
    & = \frac{1}{d} \sum_{z_{1}, z_{2} \in \mathcal{S}_{j}} \ket{\psi_{z_{1}}}\!\!\bra{\psi_{z_{2}}} \otimes  \sum_{k, \ell} \braket{\psi_{z_{1}}| k}\!\!\braket{\ell | \psi_{z_{2}}} \ket{k}\!\!\bra{\ell}_{B}   \\
    & = \frac{1}{d}  \sum_{z_{1}, z_{2} \in \mathcal{S}_{j}} \ket{\psi_{z_{1}}}\!\!\bra{\psi_{z_{2}}}_{A} \otimes  \ket{\psi_{z_{1}}}\!\!\bra{\psi_{z_{2}}}_{B}   \\
\end{align}
Suppose that Bob's strategy is then to implement the measurement $\hat{\Pi}_{B}' = \{\ket{\psi_{z}}\!\!\bra{\psi_{z}} \}|_{z}$ on system $B$. Then it follows that 
\begin{align}
    & \rho_{z}\left( \hat{\Pi}_{B}^{'}, \sum_{j} \Pi_{A}^{*(j)} \phi_{AB}^{(+)} \left(\Pi_{A}^{*(j)}\right)^{\dag} \right)  \\
    &= \frac{\sum_{z_{1}, z_{2} \in \mathcal{S}_{j}} \ket{\psi_{z_{1}}}\!\!\bra{\psi_{z_{2}}}_{A} \otimes  \ket{\psi_{z}}\!\!\braket{\psi_{z}|\psi_{z_{1}}}\!\!\braket{\psi_{z_{2}}|\psi_{z}} \!\!\bra{\psi_{z}}}{d\times  P\left(z \Big| \hat{\Pi}_{B}^{'}\right)} \\
    &= \frac{ \ket{\psi_{z}}\!\!\bra{\psi_{z}} \otimes  \ket{\psi_{z}}\!\!\bra{\psi_{z}}}{d\times  \frac{1}{d}} = \ket{\psi_{z}}\!\!\bra{\psi_{z}} \otimes \ket{\psi_{z}}\!\!\bra{\psi_{z}}
\end{align}

Finally, we have 
\begin{align}
     & R_{\mathcal{T}, p} \left( \sum_{j} \Pi_{A}^{*(j)} \phi_{AB}^{(+)} \Pi_{A}^{*(j)} \right) \leq R_{\mathcal{T}, 0} \left( \sum_{j} \Pi_{A}^{*(j)} \phi_{AB}^{(+)} \Pi_{A}^{*(j)} \right) \\
     & \leq   \sum_{z} P(z \big| \hat{\Pi}') \left\| \text{Tr}_{B}\left[\rho_{z} \! \left(\hat{\Pi}_{\text{B}}', \sum_{j} \Pi_{A}^{*(j)} \phi_{AB}^{(+)} \Pi_{A}^{*(j)} \right)\right] \right\|_{(1)} \\
      & \leq   \sum_{z=1}^{d} \frac{1}{d} \left\| \ket{\psi_{z}}\!\!\bra{\psi_{z}}_{A} \right\|_{(1)}  = 1
\end{align}
To prove the theorem statement, we now need to demonstrate that there exists some game $G$ such that $R_{G}\left(\sigma_{AB'}\right)<1$. Consider the game for which $p = 1$ and set $\mathcal{T} = \{\delta_{w=1|z'}\}$ such that  $w=1$ with certainty. If $\hat{\Pi}_{A}^{*}$  is a rank-one measurement in the basis which is mutually unbiased with respect to the eigenbasis of $\sigma_{A}$, then $\sum_{j} \Pi_{A}^{*(j)} \sigma_{AB'} \Pi_{A}^{*(j)} = \mathbf{u}_{A} \otimes \sigma_{B}$ and we have 
\begin{align}
\small
    & R_{\mathcal{T}, 1} \left( \sum_{j} \Pi_{A}^{*(j)} \sigma_{AB} \Pi_{A}^{*(j)} \right) \\
    & \geq R_{\mathcal{T}, 0} \left( \sum_{j} \Pi_{A}^{*(j)} \sigma_{AB} \Pi_{A}^{*(j)} \right) \\
     &= \max_{\hat{\Pi}_{B}^{'}}\sum_{z} P(z \big| \hat{\Pi}') \left\| \text{Tr}\left[\rho_{z} \left(\hat{\Pi}_{\text{B}}', \mathbf{u}_{A} \otimes \sigma_{B}  \right)\right]\right\|_{(1)} \\
     &= \max_{\hat{\Pi}_{B}^{'}} \sum_{z} \text{Tr}\left[\Pi_{B}^{'(z)} \sigma_{B} \right] \left\| \mathbf{u}_{A} \right\|_{(1)} = \frac{1}{d}\\
\end{align}
Given that $\frac{1}{d} <1$, the statement of the lemma follows. 
\end{proof}
Finally, we demonstrate that when system $B$ is trivial, the ordering reduces to standard majorisation between quantum states. 
\begin{lemma}
Let $\vec{\lambda}(\rho)$ denote the vector whose elements consist of the eigenvalues of $\rho$. Then for any states $\rho_{A}$ and $\sigma_{A}$, $\sigma_{A} \precsim_{(A)} \rho_{A}$ if and only if \begin{align}
    \vec{\lambda}(\sigma_{A}) \precsim \vec{\lambda}(\rho_{A})
\end{align}
where in the above, $\precsim$ represents standard majorisation.
\end{lemma}
\begin{proof}
Both $\rho_{A}$ and $\sigma_{A}$ can be viewed as bipartite states with a trivial $B$ subsystem (i.e. $|B|=1$.) Evidently, the expected reward for a given $\mathcal{T}$ and games where $p=0$ is then
\begin{align}
    R_{\mathcal{T}, 0}(\rho_{A}) &= \max_{z'}\left( \sum_{w} t_{w|z'} \left\| \rho_{A} \right\|_{(w)} \right) \\
    R_{\mathcal{T}, 0}(\sigma_{A}) &= \max_{z'}\left( \sum_{w} t_{w|z'} \left\| \sigma_{A} \right\|_{(w)} \right) \\
\end{align}
and by considering games where $w$ is fixed, the ``only if'' direction of the theorem statement follows.
It remains to show that $\vec{\lambda}(\sigma_{A}) \precsim \vec{\lambda}(\rho_{A})$ is a sufficient condition when games of chance with $p =1$ are included. In this case, the adversarial measurement $\hat{\Pi}_{A}^{*}$ will be in a basis which is mutually unbiased with respect to $\sigma_{A}$ (or respectively $\rho_{A}$). Thus, for both games, the state after the adversarial measurement is $\mathbf{u}_{A}$. Then
\begin{align}
    R_{\mathcal{T}, 1}(\rho_{A}) =  R_{\mathcal{T}, 0}(\mathbf{u}_{A}) =  R_{\mathcal{T}, 1}(\sigma_{A})
\end{align}
and so the ordering is unchanged by including games where $p>0$. The theorem statement then follows.
\end{proof}

\section{Quantum Channel Gambling Games} 

In previous work involving gambling games for classical channels~\cite{brandsen2021entropy}, the output of the classical channel was observed directly by the host. In the quantum case, where outputs are no longer necessarily orthogonal, the observation process is now represented through the projective quantum measurement $\{\Pi_{x} \}_{x=1}^{n}$. One natural extension of the classical gambling game introduced in~\cite{brandsen2021entropy} is to determine each game with a tripartite state $\rho_{ABX}$. 

The corresponding game setup is then as follows:
\begin{enumerate}
    \item The player is given knowledge of state $\rho_{ABX}$ and chooses a pre-processing channel $\mathcal{E}$ on subsystem $A$
   
    \item Channel $\mathcal{N}$ is implemented on subsystem $A$. From here on, we denote this as $\mathcal{N}\left(\mathcal{E}\left(\rho_{ABX}\right)\right)$ with the understanding that the channel action is only on system $A$.
    
    \item The player chooses a rank-one measurement on system $AB$. The player's optimal choice will always be to measure in the ordered eigenbasis of $\mathcal{N}\left(\mathcal{E}\right)$.
    
    \item The measurement outcome $y$ is obtained and the player wins with a reward of $1$ if $y \leq x$ and loses with a reward of $0$ else
\end{enumerate}
We note that the rank-one restriction above is necessary in order to prevent the player from cheating with a trivial measurement (e.g., $\{I, 0, 0, \ldots \}$) which would allow the player to win any game with any channel, thus preventing any meaningful ordering between quantum channels. This setup is depicted in Figure 1. 

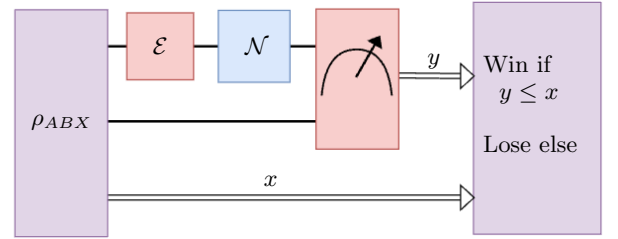
\begin{figure}[h!]
  \begin{overpic}[scale=0.62]{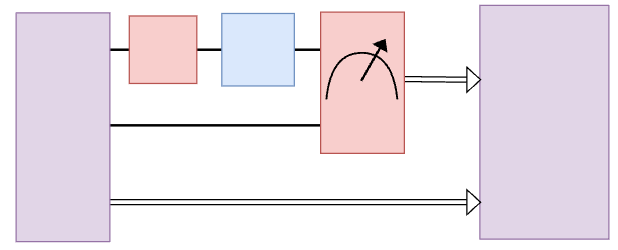}
  \put(100, 22){$x$}
  \put(12, 45){$\rho_{ABX}$}
  \put(162, 68){$y$}
  \put(58, 72){$\mathcal{E}$}
  \put(93, 72){$\mathcal{N}$}
  \put(183, 65){Win if}
  \put(186, 55){ $y \leq x$}
  \put(183, 35){Lose else}
  \large
\end{overpic}
\caption{Depiction of ``quantum gambling game''. Red corresponds to the player's choice while purple is fixed by the game.}
\end{figure}

The resulting expected reward for channel $\mathcal{N}$ and distribution $\mathbf{p}$ becomes 
\begin{align}
    R_{\rho_{ABX}}\left(\mathcal{N}\right) &= \max_{\mathcal{E}} \left( \sum_{x} p_{x} \Big\Vert \mathcal{N}\left(\mathcal{E}\left(\rho_{AB}^{\left(x\right)}\right) \Big\Vert_{\left(x\right)} \right)\right) 
\end{align}
where $\big\Vert \cdot \big\Vert_{\left(x\right)}$ is the Ky-Fan $x$-norm. 

Based on these games of chance, we can define a majorization relation between quantum channels. 
\begin{definition} 
We denote the case where channel $\mathcal{M}_{C \rightarrow D}$ is majorized by channel $\mathcal{N}_{C' \rightarrow D'}$ according to quantum games of chance as $\mathcal{M} \precsim_{q} \mathcal{N}$. Equivalently, 
\begin{align}
    \mathcal{M} \precsim_{q} \mathcal{N} \leftrightarrow R_{\rho_{ABX}} \left(\mathcal{M}\right) \leq R_{\rho_{ABX}}\left(\mathcal{N}\right) \ \ \text{for all} \ \  \rho_{ABX}
\end{align}
where it is sufficient to bound the dimensions of systems $A, B$ by $\max \left( |C|, |C'| \right)$. Likewise, the dimension of system $X$ can be bounded by $\max\left(|C|, |C'|\right) \times \min\left(|D|, |D'|\right)$. 
\end{definition}

The bound for the dimension of $X$ follows by noting that $\mathcal{N} \circ \mathcal{E}\left(\rho_{AB}\right)$ has at most $|D'| \times |B|$ nonzero eigenvalues and likewise $\mathcal{M} \circ \mathcal{E}\left(\rho_{AB}\right)$ has at most $|D| \times |B|$ nonzero eigenvalues. Since $|B| \leq \max\left(|C|, |C'|\right)$, the bound on $X$ immediately follows. 

Finally, we define the entropy of a quantum channel to be a function which is monotonic under channel ordering, namely
\begin{definition}
A non-zero function, $$H : \underset{A, B}{\bigcup} \mathrm{CPTP}(A \rightarrow B) \rightarrow \mathbb{R}$$where the union is over all finite quantum systems $A$ and $B$, is a quantum channel entropy if it satisfies the following two conditions: 
\begin{enumerate}
    \item It is monotonic under channel majorization; i.e. given quantum channels $\mathcal{N}$ and $\mathcal{M}$, then
\begin{equation}
    H(\mathcal{M}) \geq H(\mathcal{N}) \ \ \text{if} \ \ \mathcal{M} \precsim \mathcal{N}
    \end{equation}
\item It is additive under tensor products; i.e. 
\begin{equation}
H(\mathcal{N} \otimes \mathcal{M})=H(\mathcal{N})+H(\mathcal{M})
\end{equation}
for all quantum channels $\mathcal{N}$ and $\mathcal{M}$.
\end{enumerate}
\end{definition}

\section{Operational Interpretation and Results for Special Classes of Channels}

We begin by showing that when $\mathcal{M}$ and $\mathcal{N}$ correspond to states or classical channels, respectively, the quantum channel majorization defined by gambling games reduces to the standard vector majorization and the classical channel majorization introduced by our previous work~\cite{brandsen2021entropy} respectively.  

\begin{lemma} If $\mathcal{N}$ and $\mathcal{M}$ are replacement channels which output states $\rho_{\mathcal{N}}$ and $\rho_{\mathcal{M}}$ respectively, then
\begin{align}
   \mathcal{M} \precsim_{q} \mathcal{N} \ \ \ \text{if and only if} \ \ \ \vec{\lambda}\left(\rho_{\mathcal{M}} \right) \precsim \vec{\lambda}\left(\rho_{\mathcal{M}} \right)
\end{align}
where $\vec{\lambda}\left(\rho\right)$ denotes the ordered eigenvalues of $\rho$.
\end{lemma}
\begin{proof} 
See Appendix~\ref{appendixA}. 
\end{proof}

Previous work introduced classical gambling games parameterized by a correlated source $T = \left(t_{wz}\right)$ where $t_{wz}$ represents the joint probability $\operatorname{Pr}\left(W=w, Z=z\right)$. The expected reward for playing a $T$-game with classical channel $\mathcal{N}$ with transition probability matrix elements $\{p_{y|x} \}$ was likewise defined as
\begin{align}
    \operatorname{Prob}_{T}\left(\mathcal{N}\right) & \triangleq \sum_{z = 1} \max_{x} t_{w, z} \sum_{y=1}^{w} p_{y|x}^{\downarrow}
\end{align}
where $p^{\downarrow}$ indicates that $\{p_{y|x}\}|_{y}$ is ordered such that  $p_{1|x} \geq p_{2|x} \geq \ldots  p_{n|x}$ for all $x$. 

We now demonstrate that when $\mathcal{M}$ and $\mathcal{N}$ are both classical, the ordering induced by quantum gambling games is equivalent to the ordering induced by classical gambling games. 

\begin{theorem}
If $\mathcal{N}$ and $\mathcal{M}$ are both classical channels, then 
\begin{align}
   & R_{\rho_{ABX}}\left(\mathcal{M}\right)  \leq R_{\rho_{ABX}}\left(\mathcal{N}\right) \ \ \forall \ \rho_{ABX}
   \end{align}
if and only if
\begin{align}
\operatorname{Prob}_{T}\left(\mathcal{M}\right) \leq \operatorname{Prob}_{T}\left(\mathcal{N}\right) \ \ \forall \ \ T = \left(t_{wz}\right).
\end{align}
\end{theorem}

\begin{proof}
See Appendix~\ref{appendixB}. 
\end{proof} 

We now provide an operational interpretation for quantum channel majorization. 
\begin{theorem}
\label{channel_char}
If $\mathcal{M} = \sum_{z=1}^{s} p_{z} \mathcal{V}_{z} \circ \mathcal{N} \circ \mathcal{E}_{z}$ for some set of isometries $\mathcal{V}_{z}$, some set of preprocessing quantum processes $\mathcal{E}_{z}$, and some probability distribution $\{p_{1}, \ldots p_{s} \}$, then:
\begin{align}
    \mathcal{M} \precsim \mathcal{N}
\end{align}
\end{theorem}
\begin{proof}
We can demonstrate the theorem statement by substituting $\mathcal{M} = \sum_{z=1}^{s} p_{z} \mathcal{V}_{z} \circ \mathcal{N} \circ \mathcal{E}_{z}$ into the expression for the reward function as follows:
\begin{align}
    & R_{\rho_{ABX}}\left(\mathcal{M}\right) = \max_{\mathcal{E}} \left( \sum_{x} p_{x} \Big\Vert \mathcal{M} \circ \mathcal{E}\left(\rho_{AB}^{\left(x\right)}\right) \Big\Vert_{x} \right) \\
    &= \max_{\mathcal{E}} \left( \sum_{x} p_{x} \Big\Vert \sum_{z=1}^{s} p_{z} \mathcal{V}_{z} \circ \mathcal{N} \circ \mathcal{E}_{z}\left(\mathcal{E}\left(\rho_{AB}^{\left(x\right)}\right)\right) \Big\Vert_{x} \right)   \\
    & \leq \max_{\mathcal{E}, z} \left( \sum_{x} p_{x} \Big\Vert \mathcal{V}_{z} \circ \mathcal{N} \circ \mathcal{E}\left(\rho_{AB}^{\left(x\right)}\right) \Big\Vert_{x} \right)   \\
    & = \max_{\mathcal{E}, z} \left( \sum_{x} p_{x} \Big\Vert \mathcal{N} \circ \mathcal{E}\left(\rho_{AB}^{\left(x\right)}\right) \Big\Vert_{x} \right)   \\
    &= R_{\rho_{ABX}}\left(\mathcal{N}\right)  
\end{align}
This concludes the proof.
\end{proof}


It follows immediately from the above Theorem that two channels which are equivalent up to a unitary will have equal performance for any game of chance. We now show that for a fixed input and output system, unitaries are the ``least noisy'' channel. 

\begin{lemma}
 Let $\mathcal{N} _{A \rightarrow A}$ be a quantum channel, and let $\mathcal{U} _{A \rightarrow A}$ be any unitary channel. Then 
\begin{align}
    R_{\rho_{ABX}}\left(\mathcal{U}\right) \geq R_{\rho_{ABX}}\left(\mathcal{N}\right) \ \ \forall \rho_{ABX}
\end{align}
Additionally, $\mathcal{U}_{C \rightarrow \text{C}} \prec \mathcal{U}_{C' \rightarrow C'}$ if $|C| \leq |C'|$.
\end{lemma}

\begin{proof}
We first show that unitary channels outperform any other channel of the same dimension. 
\begin{align}
R_{\rho_{ABX}}\left(\mathcal{N}\right) &= \max_{\mathcal{E}} \left( \sum_{x} p_{x} \Big\Vert \mathcal{N} \circ  \mathcal{E}\left(\rho_{AB}^{\left(x\right)}\right) \Big\Vert_{x} \right)  \\
& \leq \max_{\tilde{\mathcal{E}}} \left( \sum_{x} p_{x} \Big\Vert \tilde{\mathcal{E}}\left(\rho_{AB}^{\left(x\right)}\right) \Big\Vert_{x} \right)    \\
& = \max_{\tilde{\mathcal{E}}} \left( \sum_{x} p_{x} \Big\Vert U \tilde{\mathcal{E}}\left(\rho_{AB}^{\left(x\right)}\right) U^{\dag} \Big\Vert_{x} \right)    \\
&= R_{\rho_{ABX}}\left(\mathcal{U}\right)   .
\end{align}
The statement that higher dimensional unitary channels outperform lower dimensional unitary channels follow by setting $\rho_{AB} = \ket{\Phi_{AB}^{+}}\!\!\bra{\Phi_{AB}^{+}}$ where $|A| = |B| = |C'|$ (namely, the dimensions of systems $A$ and $B$ are set to be equal to the input dimension of the higher dimensional unitary). 

Thus, $\mathcal{U}_{C' \rightarrow C'}\left(\rho_{AB}\right)$ will result in a pure state and perfect reward for any game. On the other hand, any game played with $\mathcal{U}_{C \rightarrow C}$ will require some choice of preprocessing channel $\mathcal{E}_{C' \rightarrow C}$, which is necessarily entanglement breaking. 
\end{proof}

\section{Uniqueness of the Quantum Channel Gambling Game}

In principle, one could consider more general games with additional resources. In  Appendix~C, we demonstrate that generalizations of the quantum gambling game do not yield any reasonable interpretation related to entropy. \\
First, we show that allowing the player to pass side information through a classical wire from $\mathcal{E}$ to the final measurement would enable the player to win any game with the classical identity channel. Given that quantum entanglement is a resource, we would expect the quantum identity channel to strictly outperform the classical identity channel, and so the gambling game with an added classical communication wire does not lead to a reasonable ordering between channels. 

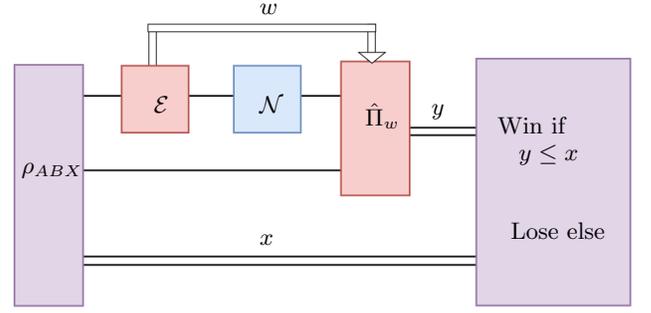
\begin{figure}[h!]
  \begin{overpic}[scale=0.42]{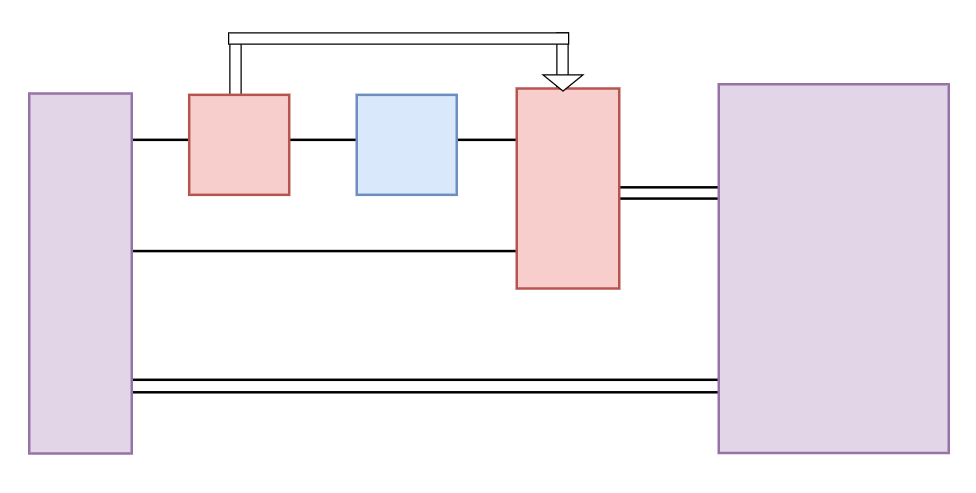}
  \put(100, 33){$x$}
  \put(10, 60){$\rho_{ABX}$}
  \put(165, 82){$y$}
  \put(60, 83){$\mathcal{E}$}
  \put(100, 83){$\mathcal{N}$}
  \put(190, 75){Win if}
  \put(195, 65){ $y \leq x$}
  \put(195, 35){Lose else}
  \put(140, 78){$\hat{\Pi}_{w}$}
  \put(100, 120){$w$}
  \large
\end{overpic}
\caption{Gambling game with added communication wire.}
\end{figure}

One could additionally consider a game with quantum combs. In Appendix~\ref{appendixD}, we discuss why allowing arbitrary quantum combs $C_{j}$ would again prevent a reasonable channel ordering. Specifically, if $\{C_{j}\}$ are allowed to be arbitrary bipartite channels, it is not possible to obtain an ordering between channels $\mathcal{N}$ and $\mathcal{U} \circ \mathcal{N}$ where $\mathcal{U}$ is any unitary and $\mathcal{N}$ is a pure state replacement channel. Given that we expect two channels which are related by a unitary to have equivalent entropy, this game must again be excluded.

\begin{figure}[h!]
  \begin{overpic}[scale=0.38]{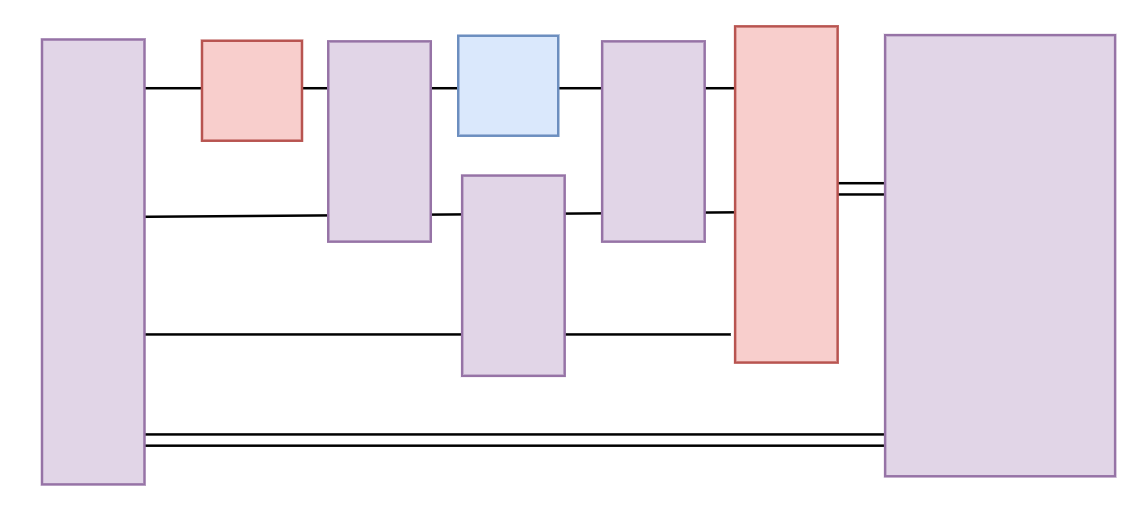}
  \small
  \put(100, 22){$x$}
  \put(113, 52){$C_{2}$}
  \put(144, 82){$C_{3}$}
  \put(82, 82){$C_{1}$}
  \put(10, 60){$\rho_{ABC}^{X}$}
  \put(195, 80){$y$}
   \put(175, 75){$\hat{\Pi}$}
  \put(55, 95){$\mathcal{E}$}
  \put(113, 95){$\mathcal{N}$}
  \put(210, 85){Win if}
  \put(215, 75){ $y \leq x$}
  \put(210, 45){Lose else}
  \large
\end{overpic}
\end{figure}

\section{Noisy Channels}
We include a table summarizing the reward function of special qubit channels for two key games involving the maximally entangled state and a pure tensor product state respectively. More specifically, in the first game we consider, we define $\rho_{AB} = \phi_{AB}^{+}$ and let $X$ be drawn according to the classical state $\rho_{X}$ s.t. $\rho_{ABX} = \phi_{AB}^{+} \otimes \rho_{X}$. In the second game we consider, system $B$ is set to be trivial such that $\rho_{AB} = \rho_{A} = \ket{0}\!\!\bra{0}$ and $\rho_{ABX} = \ket{0}\!\!\bra{0} \otimes \rho_{X}$. In both cases, we denote the ordered diagonal elements of $\rho_{X}$ as $[p_{j}]$. 
\begin{center}
 \begin{tabular}{||c c c||} 
 \hline
 $\mathcal{N}$ & $R_{\phi^{+}_{AB} \otimes \rho_{X}}\left(\mathcal{N}\right)$ & $R_{\ket{0}\!\!\bra{0} _{A} \otimes \rho_{X}}\left(\mathcal{N}\right)$ \\ [0.5ex] 
 \hline\hline
  $\mathcal{U}$ & 1 & 1 \\
   \hline
 $\mathcal{I}_{\text{CL}}$ & $1- \frac{1}{2}p_{1}$ & 1 \\
   \hline
 $D_{\gamma}$ & $\left(1-\gamma\right) +  \frac{\gamma \sum_{x} p_{x} x}{4}$ & $1-\gamma \frac{p_{1}}{2}$ \\ 
 \hline
 $\mathcal{N}_{\hat{\Pi}}$ & $1 - \frac{1}{2} p_{1}$ & 1 \\
 \hline
 $\mathcal{A}_{\gamma}$ & $(1-\frac{\gamma}{2})p_{1} + \frac{\gamma}{2}$ & 1 \\
 \hline
 $\mathcal{R}_{\sigma}$ & $\sum_{x} p_{x} \| \sigma \otimes \frac{I}{2} \|_{\left(x\right)}$ & $\sum_{x} p_{x} \| \sigma \|_{\left(x\right)}$ \\
 \hline
 $\mathcal{F}_{\gamma}$ & $\left(1-\gamma\right) + \gamma\left(1- \frac{1}{2} p_{1}\right)$ & 1 \\ [1ex] 
 \hline
\end{tabular}
\end{center}

We now define all channels included in the above table. 
\begin{enumerate}
    \item $\mathcal{U}$ is a unitary channel with corresponding unitary $U$ such that $\mathcal{U}(\rho) = U \rho U^{\dag}$ for all $\rho \in \mathcal{D}(\mathcal{H}_{2})$.
    \item $\mathcal{I}_{\text{CL}}$ is a classical identity channel such that for some basis $\{\ket{j}\}$ then
    \begin{align}
        \mathcal{I}_{\text{CL}}(\rho) = \sum_{j} \bra{j} \! \rho \! \ket{j} \ket{j}\!\!\bra{j} 
    \end{align}
    \item $\mathcal{D}_{\gamma}$ is the qubit depolarizing channel with noise parameter $\gamma$ defined as 
    \begin{align}
        \mathcal{D}_{\gamma}(\rho) = (1-\gamma)\rho + \gamma \mathbf{u}_{2}
    \end{align}
    \item $\mathcal{N}_{\hat{\Pi}}$ is the channel implementing POVM $\hat{\Pi} = \{\Pi_{j}\}$ as 
    \begin{align}
        \mathcal{N}_{\hat{\Pi}}(\rho) = \sum_{j} \text{Tr}\left[ \Pi_{j} \rho \right] \ket{j}\!\!\bra{j}
    \end{align}
    \item $\mathcal{A}_{\gamma}$ is the amplitude damping channel with noise parameter $\gamma$ defined as
    \begin{align}
    \mathcal{A}_{\gamma}\left(\begin{pmatrix}
    \rho_{00} & \rho_{01} \\
    \rho_{10} & \rho_{11}
    \end{pmatrix}\right) & \triangleq \begin{pmatrix}
    \rho_{00} + \gamma \rho_{11} & \sqrt{1-\gamma} \rho_{01} \\
    \sqrt{1-\gamma} \rho_{10} & \left(1-\gamma\right) \rho_{11}.
    \end{pmatrix}
\end{align}
\item The replacement channel $\mathcal{R}_{\sigma}$ discards the input and outputs state $\sigma$ such that $\mathcal{R}_{\sigma}(\rho) = \sigma$ for all $\rho$.
\item The dephasing channel $\mathcal{F}_{\gamma}$ with parameter $\gamma$ acts as
\begin{align}
    \mathcal{F}_{\gamma}(\rho) = (1-\gamma) \rho + \gamma I_{\text{CL}}(\rho)
\end{align}
\end{enumerate}
Details for how the results in the table are obtained are included in Appendix~\ref{Example_channels}.

Finally, we note that the maximal output state purity attainable with a channel is monotonic under the ordering induced by games of chance. 

\begin{lemma} Suppose that $\mathcal{M} \precsim \mathcal{N}$. Then \begin{align}
    \underset{\rho}{\operatorname{sup}} \left( \operatorname{Tr}\Big[ \mathcal{N}\left(\rho\right)^{2} \Big] \right) \geq   \underset{\rho}{\operatorname{sup}} \left( \operatorname{Tr}\Big[ \mathcal{M}\left(\rho\right)^{2} \Big] \right) 
\end{align}
\end{lemma}
\begin{proof} See Appendix~\ref{appendixE} 
\end{proof}



\section{Conclusions}
In this work, we demonstrate that the set of \emph{all} conditional entropy non-decreasing channels is equivalent to the set of conditionally unital, semi-causal (CUSC) channels. We then require quantum conditional entropy to be monotonic under the action of CUSC channels as well as satisfying simple axioms such as additivity and normalisation. This minimalist approach is sufficient to guarantee negativity of conditional entropy for maximally entangled states, and additionally is sufficient demonstrate that conditional entropy is non-negative for separable bipartite states. We discuss applications of this axiomatic approach to classical bipartite states and demonstrate that the resulting conditional majorization matches the classical conditional majorization first defined in~\cite{Gour}. Finally, we develop an operational approach for characterising entropy via games of chance, and apply this approach to both quantum bipartite states and quantum channels respectively.

{\it Note Added: We recently became aware of independent work of Vempati et al.~\cite{Vempati2021}}

\section{Acknowledgments}
The authors would like to thank Henry Pfister for helpful discussions. GG and IG acknowledge support from the Natural Sciences and Engineering Research Council of Canada (NSERC). SB acknowledges support from the National Science Foundation (NSF) under Grant No. 1908730 and 1910571. Any opinions, findings, conclusions, and recommendations expressed in this material are those of the authors and do not necessarily reflect the views of these sponsors. 

\bibliographystyle{unsrt}
\bibliography{ref}

\appendix

\section{Quantum Gambling Games on Bipartite Classical States}
\label{appendixC}

In this appendix, we demonstrate that the majorisation relation $\precsim_{(A)}$ is equivalent to the majorisation relation $\precsim_{C}$ on classical states. We here denote the expected success probability of playing a $\mathcal{T}$-game with classical bipartite state $\rho_{AB}$ as
\begin{align}
\operatorname{Prob}_{\mathcal{T}}\left(\rho_{AB}\right) &= \sum_{x} p_{x} \max_{z'} \left( \sum_{w} t_{w|z'} \left\| \mathbf{p}_{x} \right\|_{\left(w\right)}\right),
\end{align} 
where $\mathbf{p}_{x} = \{p_{y|x}\}_{x}$ is the probability that the output of $\rho_{AB}$ on system $B$ is $Y=y$, given that the output on system $A$ is $X=x$. \begin{lemma}
For any classical bipartite states $\rho_{AB} = \{q_{x, y}\}_{x,y}$ and $\sigma_{AB} = \{p_{x, y} \}_{x,y}$, the following inequality holds
\begin{align}
 \operatorname{Prob}_{\mathcal{T}}\left(\sigma_{AB}\right) \leq \operatorname{Prob}_{\mathcal{T}}\left(\rho_{AB}\right) \ \forall \mathcal{T}
 \end{align} 
 if and only if 
 \begin{align}
  R_{\mathcal{T}, p}\left(\sigma_{AB}\right) \leq R_{\mathcal{T}, p}\left(\rho_{AB}\right) 
 \end{align}
for all distributions $\mathcal{T}$ and all $0 \leq p \leq 1$
\end{lemma}
\begin{proof}
First, we show that for all quantum gambling games with distribution $\mathcal{T}$ and $p =0$, the following equality holds 
\begin{align}
    R_{\mathcal{T}, 0}\left(\sigma_{AB}\right) = \operatorname{Prob}_{\mathcal{T}}\left(\sigma_{AB}\right).
\end{align}
Recall that any classical bipartite state $\rho_{AB}$ can be written in quantum notation as 
 \begin{align}
     \rho_{AB} &= \sum_{x, y} q_{x, y} \ket{x}\!\!\bra{x} \otimes \ket{y}\!\!\bra{y}, \\
     \sigma_{AB} &= \sum_{x, y} p_{x, y} \ket{x}\!\!\bra{x} \otimes \ket{y}\!\!\bra{y}.
 \end{align}
First consider games where $p = 0$. Then 
\begin{widetext}
\begin{align}
    R_{\mathcal{T}, 0}\left(\rho_{AB}\right) &=
     \underset{\hat{\Pi}^{'},  f }{\text{max}} \left( \sum_{z} \sum_{w} t_{w|f(z)} \left\|  \text{Tr}_{B}\left[\Pi_{\text{B}}^{\prime (z)}  \rho_{AB} \Pi_{\text{B}}^{'(z)}\right\|_{(w)}\right] \right) \\
     &= 
     \underset{\hat{\Pi}^{'},  f }{\text{max}} \left( \sum_{z} \sum_{w} t_{w|f(z)} \left\| \sum_{x}  q_{x, y} \ket{x}\!\!\bra{x} \text{Tr}\left[ \Pi_{\text{B}}^{\prime (z)} \ket{y}\!\!\bra{y}  \Pi_{\text{B}}^{'(z)}\right]\right\|_{(w)} \right) \\
     &= 
     \underset{f }{\text{max}} \left( \sum_{z} \sum_{w} t_{w|f(z)} \left\| \sum_{x}  q_{x, y} \ket{x}\!\!\bra{x} \text{Tr}\left[ \ket{z}\!\! \braket{z|y}\!\!\braket{y|z}\!\!\bra{z}  \right]\right\|_{(w)} \right) \\
     &= 
     \underset{f }{\text{max}} \left(\sum_{z} \sum_{w} t_{w|f(z)} \left\| \sum_{x}  q_{x, z} \ket{x}\!\!\bra{x} \right\|_{(w)} \right) \\
\end{align}
\end{widetext}
Upon relabelling variables $x$ to $k$ for consistency with the convention used for classical gambling games, we have
\begin{align}
 R_{\mathcal{T}, 0}(\rho_{AB})
&=  \sum_{z} \max_{z'} \sum_{w} t_{w|z'} \Big \| \sum_{k}  q_{k, z} \ket{k}\!\!\bra{k}  \Big \|_{\left(w\right)}   \\
&=  \sum_{z} q_{z} \max_{z'} \sum_{w} t_{w|z'} \Big \| \sum_{k}  q_{k|z} \ket{k}\!\!\bra{k}  \Big \|_{\left(w\right)}   \\
&=  \sum_{z} q_{z} \max_{z'} \sum_{w} t_{w|z'} \Big \| \mathbf{q}_{z} \Big \|_{\left(w\right)}   \\
&= \operatorname{Prob}_{\mathcal{T}}\left(\rho_{AB}\right)  .
\end{align}
It immediately follows that $\operatorname{Prob}_{\mathcal{T}}\left(\sigma_{AB}\right) \leq \operatorname{Prob}_{\mathcal{T}}\left(\rho_{AB}\right)$ for all $\mathcal{T}$ if $R_{\mathcal{T}, p}\left(\sigma_{AB}\right) \leq R_{\mathcal{T}, p}\left(\rho_{AB}\right)$ for all distributions $\mathcal{T}$ and all $p \in [0, 1]$.

Now, we demonstrate that the additional games allowed in the quantum case (i.e., games where $p > 1$) does not affect the ordering. This follows immediately upon noting that the adversary can measure in a mutually unbiased basis, such that the state $\rho_{AB}$ transforms to $\mathbf{u}_{A} \otimes \sigma_{B}$ for some $\sigma_{B}$. 

It follows that $R_{\mathcal{T}, 1}\left(\rho_{AB}\right) = R_{\mathcal{T}, 1}\left(\sigma_{AB}\right)$ for all $\mathcal{T}$ and therefore the additional games do affect the partial ordering. 
\end{proof}

As an aside, we note that it is sufficient to consider a restricted set of games for ordering between two classical states. 

\begin{lemma}
If $\rho_{AB}$ and $\sigma_{AB}$ are two classical states, it is sufficient to restrict to games where $w$ is determined. Namely, 
\begin{align}
    \operatorname{Prob}_{G}\left(\sigma_{AB}\right) \leq \operatorname{Prob}_{G}\left(\rho_{AB}\right)
\end{align}
for all games $G$ if and only if
\begin{align}
    \operatorname{Prob}_{(w, \mathbf{p})}(\sigma_{AB}) \leq \operatorname{Prob}_{\left(w, \mathbf{p}\right)}\left(\rho_{AB}\right),
\end{align}
for all $G$ where $w$ is fixed (i.e.$\mathcal{T} = \{\delta_{w|z'}\}|_{z'}$) and where $\mathbf{p} = \{1,0\}$, so that an adversarial measurement never occurs. 
\end{lemma}

\begin{proof}
It follows immediately from the previous proof that it is sufficient to restrict to games where $\mathbf{p} = \{1,0\}$, as for an arbitrary $\mathbf{p} = \{p, 1-p\}$, the reward for any bipartite state $\rho_{AB}$ can be written as
\begin{align}
 R_{\left(\mathcal{T}, \{p, 1-p\}\right)}\left(\rho_{AB}\right) &= p  R_{\left(\mathcal{T}, \{1, 0 \}\right)}\left(\rho_{AB}\right) \\
 & + \left(1-p\right) R_{\left(\mathcal{T}, \{0, 1 \}\right)}\left(\rho_{AB}\right).
\end{align}
We demonstrated in the previous case that for a classical state $\sigma_{AB} = \sum_{x,y} p_{x,y} \ket{x}\!\!\bra{x} \otimes \ket{y}\!\!\bra{y}$, the following equality holds 
\begin{align}
    R_{\left(\mathcal{T}, \{1, 0 \}\right)}\left(\sigma_{AB}\right) = \operatorname{Prob}_{T}\left(\sigma_{AB}\right)
\end{align}
Finally, we demonstrate that this can be rewritten as a convex combination of games where $w$ is fixed:
\begin{align}
\operatorname{Prob}_{\mathcal{T}}\left(\sigma_{AB}\right) &= \sum_{x} p_{x} \max_{z} \left( \sum_{w} t_{w|z'} \| \mathbf{p}_{x} \|_{\left(w\right)}\right) \\
&= \sum_{x} p_{x} \left( \sum_{w} t_{w|z_{x}^{*}} \| \mathbf{p}_{x} \|_{\left(w\right)}\right) \\
&= \sum_{w} \sum_{x} t_{w|z_{x}^{*}} p_{x}  \| \mathbf{p}_{x} \|_{\left(w\right)}  \\
&= \sum_{w} c_{w} \sum_{x} p_{x}  \| \mathbf{p}_{x} \|_{\left(w\right)}  ,
\end{align}
where for each $x$, the variable $z\left(x\right)^{*}$ is defined to be the optimal choice and  for each $w$,
\begin{align}
    c_{w} = \frac{\sum_{x} t_{w|z_{x}^{*}} p_{x}  \| \mathbf{p}_{x} \|_{\left(w\right)}}{\sum_{x} p_{x}  \| \mathbf{p}_{x} \|_{\left(w\right)}}
\end{align}
is a positive constant. Finally, if $\rho_{AB}$ and $\sigma_{AB}$ are ordered for all games with a fixed $w$, then
\begin{align}
& \operatorname{Prob}_{\mathcal{T}}\left(\sigma_{AB}\right) \notag \\
&= \sum_{w} c_{w} \operatorname{Prob}_{w}\left(\sigma_{AB}\right) \\
& \leq \sum_{w} c_{w} \operatorname{Prob}_{w}\left(\rho_{AB}\right) \\
& = \sum_{w} \frac{\sum_{x} t_{w|z_{x}^{*}} p_{x}  \| \mathbf{p}_{x} \|_{\left(w\right)}}{\sum_{x} p_{x}  \| \mathbf{p}_{x} \|_{\left(w\right)}} \sum_{x} q_{x}  \| \mathbf{q}_{x} \|_{\left(w\right)} \\
& = \sum_{x}  q_{x} \sum_{w} \left( \sum_{x'} \tilde{p}_{x'} t_{w|z_{x'}^{*}} \right)  \| \mathbf{q}_{x} \|_{\left(w\right)} \\
& \leq \sum_{x}  q_{x} \max_{x'} \left( \sum_{w} t_{w|z_{x'}^{*}} \| \mathbf{q}_{x} \|_{\left(w\right)} \right) \\
& \leq \sum_{x}  q_{x} \max_{z'} \left( \sum_{w} t_{w|z'} \| \mathbf{q}_{x} \|_{\left(w\right)} \right) \\
&= \operatorname{Prob}_{\mathcal{T}}\left(\rho_{AB}\right),
\end{align}
where $\tilde{p}_{x} =  \frac{p_{x}  \| \mathbf{p}_{x} \|_{\left(w\right)}}{\sum_{x} p_{x}  \| \mathbf{p}_{x} \|_{\left(w\right)}}$. 
\end{proof}

\section{Proof of reduction for quantum replacement channels}
\label{appendixA}
The ``only if'' part follows trivially. We now prove the ``if'' part of the statement. First, recall that in~\cite{brandsen2021entropy}, it was demonstrated that for classical games of chance $\mathcal{M} \prec \mathcal{N}$ if and only if $\rho_{\mathcal{M}} \precsim \rho_{\mathcal{N}}$. Additionally, classical games of chance can be represented as quantum games of chance where $\rho_{ABX} = \ket{0}\!\!\bra{0}_{A} \otimes \rho_{X}$. Thus, the ``if'' direction of the theorem statement may be rewritten as

\begin{align}
   & R_{\rho_{ABX}}\left(\mathcal{M}\right)  \leq R_{\rho_{ABX}}\left(\mathcal{N}\right) \ \ \forall \ \rho_{ABX} 
\end{align}
if
\begin{align}
   & R_{\rho_{AX}}\left(\mathcal{M}\right)  \leq R_{\rho_{AX}}\left(\mathcal{N}\right) \ \ \forall \ \rho_{ABX} = \ket{0}\!\!\bra{0}_{A} \otimes \rho_{X}
   \end{align}

Clearly, since $\mathcal{M}$ is a replacement channel, then $\mathcal{M} \circ \mathcal{E} = \mathcal{M}$ for all $\mathcal{E}$. Hence, the choice of preprocessing channel does not affect the reward, so w.l.o.g. we may set $\mathcal{E} = \mathcal{I}$. It follows that the reward then becomes 
\begin{align}
    R_{\rho_{ABX}}\left(\mathcal{M}\right) &= \max_{\mathcal{E}} \left( \sum_{x=1}^{\ell} p_{x} \Big\Vert \mathcal{M} \circ \mathcal{E}\left(\rho_{AB}^{\left(x\right)}\right) \Big\Vert_{x} \right) \\
    &= \sum_{x=1}^{\ell} p_{x} \Big\Vert \mathcal{M}\left(\rho_{AB}^{\left(x\right)}\right) \Big\Vert_{x} \\
    &= \sum_{x=1}^{\ell} p_{x} \Big\Vert \rho_{M} \otimes \operatorname{Tr}_{A}[\rho_{AB}^{\left(x\right)}] \Big\Vert_{x}
\end{align}
 We now denote $\mathbf{\lambda}^{\downarrow}\left(\rho_{\mathcal{M}}\right) = [\lambda_{1}^{M}, \lambda_{2}^{M}, \ldots , \lambda_{d}^{M}]$ and likewise $\mathbf{\lambda}^{\downarrow}\left(\operatorname{Tr}_{A}[\rho_{AB}^{\left(x\right)}]\right) = [\lambda_{1, x}^{B}, \lambda_{2, x}^{B}, \ldots , \lambda_{d, x}^{B}]$. Then
\begin{align}
    \text{eig}\left( \rho_{\mathcal{M}} \otimes \operatorname{Tr}_{A}[\rho_{AB}^{\left(x\right)}] \right) &= \{ \lambda_{j}^{\left(M\right)} \lambda_{k, x}^{B} \}_{j, k} 
\end{align}
We define $f\left(j, k, x\right)$ to return the position of $\lambda_{j}^{\left(M\right)} \lambda_{k, x}^{B}$ in the full ordered set $\{ \lambda_{j}^{\left(M\right)} \lambda_{k, x}^{B} \}_{j, k}$ (for example, $f\left(1, 1, x\right) = 1$ since $\lambda_{1}^{\left(M\right)} \lambda_{1, x}^{B}$ is the largest eigenvalue and $f\left(d, d, x\right) = d^{2}$ since $\lambda_{d}^{\left(M\right)} \lambda_{d, x}^{B}$ is the smallest eigenvalue). Finally, we can define the set 
\begin{align}
    S\left(j\right) = \{x, k \ \big| \ \ \text{such that } \  f\left(j, k, x\right) \leq x \}
\end{align}
In other words, $S\left(j\right)$ provides all values of $x, k$ such that  the term $p_{x} \lambda_{j}^{\left(M\right)} \lambda_{k, x}^{B}$ appears in the sum. 

 We now rewrite the reward function as
\begin{align}
    R_{\rho_{ABX}}\left(\mathcal{M}\right) &= \sum_{x=1}^{\ell} p_{x} \sum_{j, k \ \text{s.t.} \ f\left(j, k, x\right) \leq x} \lambda_{j}^{\left(m\right)} \lambda_{k, x}^{B}\\
    &= \sum_{j=1}^{d} \lambda_{j}^{\left(M\right)} \left( \sum_{x, j \in S\left(j\right)}  p_{j} \lambda_{k, x}^{B} \right) 
\end{align}
In previous work, we proved that games of the form $\sum_{j} \lambda_{j}^{\left(M\right)} C_{j}$ for $C_{j} > 0$ can be rewritten in the form $\alpha R_{\tilde{p}}\left(\mathcal{M}\right)$. In this case, evidently 
\begin{align}
    C_{j} \triangleq \underset{{x, k \in S\left(j\right)}}{\sum}  p_{x} \lambda_{k, x}^{B} > 0
\end{align} 
so there exists some $\alpha$ and $\tilde{p}$ such that  $R_{\rho_{ABX}}\left(\mathcal{M}\right) = \alpha R_{\tilde{p}}\left(\mathcal{M}\right)$. 

Since we are assuming that $\mathcal{M} \precsim \mathcal{N}$, then $R_{\tilde{p}}\left(\mathcal{M}\right) \leq R_{\tilde{p}}\left(\mathcal{N}\right)$. It follows that

Finally, we have 
\begin{align}
    R_{\rho_{ABX}}\left(\mathcal{M}\right) & = \alpha R_{\tilde{p}}\left(\mathcal{M}\right) \\
    & \leq \alpha R_{\tilde{p}}\left(\mathcal{N}\right) \\
    & =  \max_{z} \left( \sum_{k=1}^{d} \lambda_{k}^{\downarrow}\left( \rho_{\mathcal{N}} \right) \left( \sum_{j=1}^{d} \lambda_{j}^{B} \sum_{\ell = f_{k}\left(j\right)}^{d^{2}} p_{\ell} \right) \right) \\
    & \leq \max_{z} \left( \sum_{x} p_{x} \Big\Vert \rho_{\mathcal{N}} \otimes \operatorname{Tr}_{A}[\rho_{AB}^{\left(x\right)}] \Big\Vert\right) \\
    & = R_{\rho_{ABX}}\left(\mathcal{N}\right)
\end{align}
The statement then follows from recalling that classical gambling games are equivalent to quantum gambling games where $\rho_{ABX}$ has the form $\ket{0}\!\!\bra{0}_{A} \otimes \rho_{X}$.

\section{Proof of reduction for Classical Channels}
\label{appendixB}
The ``only if'' part follows immediately from noting that the quantum gambling games are a strict generalisation of the classical gambling games. We now prove the ``if'' part of the statement. For $\mathcal{M}$ classical, $\mathcal{E}$ can be reduced to selecting the optimal classical output $z$, such that  $\mathcal{E}\left(\rho\right) = \ket{z}\!\!\bra{z}$ for every $\rho$. It follows that the reward then becomes 
\begin{align}
    R_{\rho_{ABX}}\left(\mathcal{M}\right) &= \max_{\mathcal{E}} \left( \sum_{x=1}^{\ell} p_{x} \Big\Vert \mathcal{M}\left(\mathcal{E}\left(\rho_{AB}^{\left(x\right)}\right)\right) \Big\Vert_{x} \right) \\
    &= \max_{z} \left( \sum_{x=1}^{\ell} p_{x} \Big\Vert \mathcal{M}\left(\ket{x}\!\!\bra{x}\right) \otimes \operatorname{Tr}_{A}[\rho_{AB}^{\left(x\right)}] \Big\Vert_{x} \right) \\
    &= \sum_{x=1}^{\ell} p_{x} \Big\Vert \lambda_{j}^{\downarrow}\left( \mathcal{M}\left(\ket{z^{*}}\!\!\bra{z^{*}}\right) \right) \otimes \operatorname{Tr}_{A}[\rho_{AB}^{\left(x\right)}] \Big\Vert_{x}
\end{align}
where $z^{*}$ is the optimal value (i.e. the value of $z$ which maximises the reward). We now denote $\mathbf{\lambda}^{\downarrow}\left(\mathcal{M}\left(\ket{z^{*}}\!\!\bra{z^{*}}\right)\right) = [\lambda_{1}^{M}, \lambda_{2}^{M}, \ldots , \lambda_{d}^{M}]$ and likewise $\mathbf{\lambda}^{\downarrow}\left(\operatorname{Tr}_{A}[\rho_{AB}^{\left(x\right)}]\right) = [\lambda_{1, x}^{B}, \lambda_{2, x}^{B}, \ldots , \lambda_{d, x}^{B}]$. Then
\begin{align}
    \text{eig}\left( \mathcal{M}\left(\ket{z^{*}}\!\!\bra{z^{*}}\right) \otimes \operatorname{Tr}_{A}[\rho_{AB}^{\left(x\right)}] \right) &= \{ \lambda_{j}^{\left(M\right)} \lambda_{k, x}^{B} \}_{j, k} 
\end{align}
We define $f\left(j, k, x\right)$ to return the position of $\lambda_{j}^{\left(M\right)} \lambda_{k, x}^{B}$ in the full ordered set $\{ \lambda_{j}^{\left(M\right)} \lambda_{k, x}^{B} \}_{j, k}$ (for example, $f\left(1, 1, x\right) = 1$ since $\lambda_{1}^{\left(M\right)} \lambda_{1, x}^{B}$ is the largest eigenvalue and $f\left(d, d, x\right) = d^{2}$ since $\lambda_{d}^{\left(M\right)} \lambda_{d, x}^{B}$ is the smallest eigenvalue). Finally, we can define the set 
\begin{align}
    S\left(j\right) = \{x, k \ \big| \ \ \text{such that } \  f\left(j, k, x\right) \leq x \}
\end{align}
In other words, $S\left(j\right)$ provides all values of $x, k$ such that  the term $p_{x} \lambda_{j}^{\left(M\right)} \lambda_{k, x}^{B}$ appears in the sum. \\

   We now rewrite the reward function as
\begin{align}
    R_{\rho_{ABX}}\left(\mathcal{M}\right) &= \sum_{j=1}^{d} \lambda_{j}^{\left(M\right)} \left( \sum_{x, j \in S\left(j\right)}  p_{j} \lambda_{k, x}^{B} \right) 
\end{align}
However, this is of the form of a classical gambling game (possibly up to a scaling). In previous work, we proved that games of the form $\sum_{j} \lambda_{j}^{\left(M\right)} C_{j}$ for $C_{j} > 0$ can be rewritten in the form $\alpha R_{\tilde{p}}\left(\mathcal{M}\right)$. In this case, evidently 
\begin{align}
    C_{j} \triangleq \underset{{x, k \in S\left(j\right)}}{\sum}  p_{x} \lambda_{k, x}^{B} > 0
\end{align} 
so there exists some $\alpha$ and $\tilde{p}$ such that  $R_{\rho_{ABX}}\left(\mathcal{M}\right) = \alpha R_{\tilde{p}}\left(\mathcal{M}\right)$. Since we are assuming that $\mathcal{M} \precsim \mathcal{N}$, then $R_{\tilde{p}}\left(\mathcal{M}\right) \leq R_{\tilde{p}}\left(\mathcal{N}\right)$. It follows that

Finally, we have 
\begin{align}
    R_{\rho_{ABX}} \left(\mathcal{M}\right) &= \alpha R_{\tilde{p}}\left(\mathcal{M}\right) \\
    & \leq \alpha R_{\tilde{p}}\left(\mathcal{N}\right) \\
    & =  \max_{z} \left( \sum_{k=1}^{d} \lambda_{k}^{\downarrow}\left(\mathcal{N}\left(\ket{z}\!\!\bra{z}\right) \right) \sum_{j=1}^{d} \lambda_{j}^{B} \sum_{\ell = f_{k}\left(j\right)}^{d^{2}} p_{\ell}  \right) \\
    & \leq \max_{z} \left( \sum_{x} p_{x} \Big\Vert \mathcal{N}\left(\ket{z}\!\!\bra{z}\right) \otimes \operatorname{Tr}_{A}[\rho_{AB}^{\left(x\right)}]\right) \\
    & = R_{\rho_{ABX}}\left(\mathcal{N}\right)
\end{align}

\section{Games with Wires and Resources}
\label{appendixD}
\subsection{Game with Classical Wire}
We first consider the game with an additional classical wire, and demonstrate that one can win any game of the form $\rho_{ABX} = \ket{\psi_{AB}}\!\!\bra{\psi_{AB}} \otimes \rho_{X}$ with a qubit classical identity channel. Define $\mathcal{E}$ as
\begin{align}
    \mathcal{E}\left(\rho\right) &= \operatorname{Tr}[\ket{0}\!\!\bra{0} \rho] \ket{0}\!\!\bra{0} + \operatorname{Tr}[\ket{1}\!\!\bra{1} \rho] \ket{1}\!\!\bra{1}
\end{align}
Thus, $\mathcal{E}$ is a POVM and if outcome $0$ (respectively 1) is obtained the post measurement state is, 
\begin{align}
    \rho_{AB}^{\left(0\right)} & \triangleq \frac{\ket{0}\!\!\bra{0}_{A} \otimes I_{B} \left( \ket{\psi_{AB}}\!\!\bra{\psi_{AB}} \right) \ket{0}\!\!\bra{0}_{A} \otimes I_{B}}{ \operatorname{Tr}[\ket{0}\!\!\bra{0}_{A} \otimes I_{B} \left( \ket{\psi_{AB}}\!\!\bra{\psi_{AB}} \right)]} \\
    \rho_{AB}^{\left(1\right)} & \triangleq \frac{\ket{1}\!\!\bra{1}_{A} \otimes I_{B} \left( \ket{\psi_{AB}}\!\!\bra{\psi_{AB}} \right) \ket{1}\!\!\bra{1}_{A} \otimes I_{B}}{ \operatorname{Tr}[\ket{1}\!\!\bra{1}_{A} \otimes I_{B} \left( \ket{\psi_{AB}}\!\!\bra{\psi_{AB}} \right)]}
\end{align}
The most general expression for $\ket{\psi_{AB}}$ is
\begin{align}
    & \ket{\psi_{AB}} = \alpha_{00} \ket{00} + \alpha_{01} \ket{01} + \alpha_{10} \ket{10} + \alpha_{11} \ket{11} \\
    & \rightarrow \ket{0}\!\!\bra{0}_{A} \otimes I_{B} \ket{\psi_{AB}} = \alpha_{00} \ket{00} + \alpha_{01} \ket{01} \\
& \rightarrow \ket{1}\!\!\bra{1}_{A} \otimes I_{B} \ket{\psi_{AB}} = \alpha_{10} \ket{10} + \alpha_{11} \ket{11}
\end{align}
from which it follows that 
\begin{align}
    \rho_{AB}^{\left(0\right)} &= \ket{0}\!\!\bra{0}_{A} \otimes \ket{\psi^{\left(0\right)}}\!\!\bra{\psi^{\left(0\right)}}_{B} \\
     \rho_{AB}^{\left(1\right)} &= \ket{1}\!\!\bra{1}_{A} \otimes \ket{\psi^{\left(1\right)}}\!\!\bra{\psi^{(1)}}_{B}
\end{align}

We now show that the following measurements will allow the player to always win the game with $\rho_{AB}^{\left(0\right)}$ and $\rho_{AB}^{\left(1\right)}$ respectively:
\begin{align}
    \hat{\Pi}^{\left(0\right)} &= \Big\{ \ket{0}\!\!\bra{0} \otimes \ket{\psi^{\left(0\right)}}\!\!\bra{\psi^{\left(0\right)}}, \ket{0}\!\!\bra{0} \otimes \ket{\psi^{\left(0\right), \perp}}\!\!\bra{\psi^{\left(0\right), \perp}}, \\
    & \ket{1}\!\!\bra{1} \otimes \ket{\psi^{\left(0\right)}}\!\!\bra{\psi^{\left(0\right)}}, \ket{1}\!\!\bra{1} \otimes \ket{\psi^{\left(0\right), \perp}}\!\!\bra{\psi^{\left(0\right), \perp}} \Big\} \\
    \hat{\Pi}^{\left(1\right)} &= \Big\{ \ket{1}\!\!\bra{1} \otimes \ket{\psi^{\left(1\right)}}\!\!\bra{\psi^{\left(1\right)}}, \ket{1}\!\!\bra{1} \otimes \ket{\psi^{\left(1\right), \perp}}\!\!\bra{\psi^{\left(1\right), \perp}} \\
    & \ket{0}\!\!\bra{0} \otimes \ket{\psi^{\left(1\right)}}\!\!\bra{\psi^{\left(1\right)}}, \ket{0}\!\!\bra{0} \otimes \ket{\psi^{\left(1\right), \perp}}\!\!\bra{\psi^{\left(1\right), \perp}} \Big\}
\end{align}

It is sufficient to show that the above measurements allow the player to always win the game where $\text{diag}\left(\rho_{X}\right)= [1,0,\ldots,0]$ (namely, $x$ is always equal to one.) The corresponding reward function for the given choice of $\mathcal{E}$ then becomes:  
\begin{align}
    \text{R}\left(I_{CL}\right) &= \sum_{x \in \{0, 1\}} \operatorname{Tr}[\ket{x}\!\!\bra{x}_{A} \rho_{AB}] \\
    & \times \operatorname{Tr}\Big[ \left( \ket{x}\!\!\bra{x} \otimes \ket{\psi^{\left(x\right)}}\!\!\bra{\psi^{\left(x\right)}} \right) \rho_{AB}^{\left(1\right)} \Big] \\
    &= \sum_{x \in \{0, 1\}} \operatorname{Tr}[\ket{x}\!\!\bra{x}_{A} \rho_{AB}] \\
    & = 1
\end{align}

\subsection{Game with Quantum Combs}
We now consider the quantum comb-based game depicted in Figure (reference), and demonstrate that allowing any element $C_{j}$ to be an arbitrary quantum channel would result in a channel ordering that cannot correspond to channel entropy. \\
First, we demonstrate that $C_{3}$ cannot be an arbitrary bipartite channel. To do so, consider the following channels:
\begin{align}
   \mathcal{N}_{1}\left(\rho\right) &= \ket{0}\!\!\bra{0} \ \ \forall \rho, \ \ \ \mathcal{N}_{2}\left(\rho\right) = \ket{+}\!\!\bra{+} \ \ \forall \rho \\
   \end{align}
We now consider the performance of both channels in the game where $\rho_{ABC}^{\left(x\right)} = \ket{000}\!\!\bra{000}$ for all $x$, where combs $C_{1}$ and $C_{2}$ are taken to be trivial (identity) channels, and where
   \begin{align}
   C_{3}\left(\rho_{AB}\right) &= \operatorname{Tr}[\ket{0}\!\!\bra{0}_{A} \rho_{AB}] \ket{00}\!\!\bra{00} + \operatorname{Tr}[\ket{1}\!\!\bra{1}_{A} \rho_{AB}] \ket{11}\!\!\bra{11}
\end{align}
Additionally, suppose that $X$ is such that $x=1$ with certainty, or equivalently, the player wins only if they always obtain the first measurement outcomes. Evidently, in the above, $\mathcal{N}_{1}$ and $\mathcal{N}_{2}$ are equivalent pure state replacement channels (up to a unitary). However, 
\begin{align}
    R\left(\mathcal{N}_{1}\right) &= \Big\Vert C_{3} \left( \mathcal{N}_{1}\left(\ket{0}\!\!\bra{0}\right)_{A} \otimes \ket{00}\!\!\bra{00}_{BC} \right) \Big\Vert_{\left(1\right)} \\
    &= \left\Vert C_{3} \left(  \ket{000}\!\!\bra{000}_{ABC} \right) \right\Vert_{\left(1\right)} \\
    &= \Big\Vert \operatorname{Tr}[\ket{0}\!\!\bra{0} \ket{0}\!\!\bra{0}] \ket{00}\!\!\bra{00}_{AB} \otimes \ket{0}\!\!\bra{0}_{C} \\
    & +  \operatorname{Tr}[\ket{1}\!\!\bra{1} \ket{0}\!\!\bra{0}] \ket{11}\!\!\bra{11}_{AB} \otimes \ket{0}\!\!\bra{0}_{C}\Big\Vert_{\left(1\right)} \\
    &= \Big\Vert \operatorname{Tr}[\ket{000}\!\!\bra{000}_{ABC} \Big\Vert_{\left(1\right)} \\
    &= 1
\end{align}
whereas
\begin{align}
    R\left(\mathcal{N}_{2}\right) &= \Big\Vert C_{3} \left( \mathcal{N}_{2}\left(\ket{0}\!\!\bra{0}\right)_{A} \otimes \ket{00}\!\!\bra{00}_{BC} \right) \Big\Vert_{\left(1\right)} \\
    &= \Big\Vert C_{3} \left(  \ket{+00}\!\!\bra{+00}_{ABC} \right) \Big\Vert_{\left(1\right)} \\
    &= \Big\Vert \operatorname{Tr}[\ket{0}\!\!\bra{0} \ket{+}\!\!\bra{+}] \ket{00}\!\!\bra{00}_{AB} \otimes \ket{0}\!\!\bra{0}_{C} \\
    & +  \operatorname{Tr}[\ket{1}\!\!\bra{1} \ket{+}\!\!\bra{+}] \ket{11}\!\!\bra{11}_{AB} \otimes \ket{0}\!\!\bra{0}_{C}\Big\Vert_{\left(1\right)} \\
    &= \left\Vert \operatorname{Tr}\left[\frac{1}{2} \ket{000}\!\!\bra{000}_{ABC} + \frac{1}{2} \ket{110}\!\!\bra{110}_{ABC}\right]\right\Vert_{\left(1\right)} \\
    &= \frac{1}{2}
\end{align}
Thus, there is no ordering between channels that are equivalent up to a unitary. 

We now similarly demonstrate that $C_{1}$ cannot be an arbitrary bipartite channel. We consider the following channels and games:
\begin{align}
   \mathcal{N}_{1}\left(\rho\right) &= \operatorname{Tr}\Big[\ket{0}\!\!\bra{0} \rho_{A}\Big] \ket{0}\!\!\bra{0}  + \operatorname{Tr}\Big[\ket{1}\!\!\bra{1} \rho_{A}\Big] \ket{1}\!\!\bra{1}  \  \\
   \mathcal{N}_{2}\left(\rho\right) &= \operatorname{Tr}\Big[\ket{+}\!\!\bra{+} \rho_{A}\Big] \ket{+}\!\!\bra{+}  \\
   & + \operatorname{Tr}\Big[\ket{-}\!\!\bra{-} \rho_{A}\Big] \ket{-}\!\!\bra{-}  \\ 
   \rho_{ABC}^{\left(x\right)} &= \ket{000}\!\!\bra{000} \ \ \forall x \\
   C_{1}\left(\rho_{AB}\right) &= \ket{00}\!\!\bra{00}_{AB} \ \ \forall \rho_{AB} \\
\end{align}
and where $C_{2}$ and $C_{3}$ are taken again to be the identity. Additionally, suppose that $X=1$ so the player wins only if they always obtain the first measurement outcomes. Evidently, in the above, $\mathcal{N}_{1}$ and $\mathcal{N}_{2}$ are both POVMS. However, 
\begin{align}
    R\left(\mathcal{N}_{1}\right) &= \Big\Vert  \mathcal{N}_{1} \circ C_{1}\left(\ket{00}\!\!\bra{00}\right)_{AB} \otimes \ket{0}\!\!\bra{0}_{C} \Big\Vert_{\left(1\right)} \\
    &= \Big\Vert \mathcal{N}_{1} \left(  \ket{0}\!\!\bra{0}_{A} \right) \otimes \ket{00}\!\!\bra{00}_{BC} \Big\Vert_{\left(1\right)} \\
    &= \Big\Vert  \ket{000}\!\!\bra{000}_{ABC} \Big\Vert_{\left(1\right)} \\
    &=  1
\end{align}
whereas
\begin{align}
    R\left(\mathcal{N}_{2}\right) &= \Big\Vert  \mathcal{N}_{2} \circ C_{1}\left(\ket{00}\!\!\bra{00}\right)_{AB} \otimes \ket{0}\!\!\bra{0}_{C} \Big\Vert_{\left(1\right)} \\
    &= \Big\Vert \mathcal{N}_{2} \left(  \ket{0}\!\!\bra{0}_{A} \right) \otimes \ket{00}\!\!\bra{00}_{BC} \Big\Vert_{\left(1\right)} \\
    &= \Big\Vert \frac{I}{2}_{A} \otimes \ket{00}\!\!\bra{00}_{BC} \Big\Vert_{\left(1\right)} \\
    &=  \frac{1}{2}
\end{align}
Thus, there is no ordering between channels that are equivalent up to a unitary.

From the above, evidently $C_{1}$ and $C_{3}$ need to be restricted. If we suppose that $C_{1}$ and $C_{3}$ are restricted to be the identity (or a unitary), then the remaining nontrivial channel is $C_{2}$.  
However, $C_{2}$ can be absorbed into the initial state preparation, so 
\begin{align}
    \rho_{ABC}^{\left(x\right)} \rightarrow C_{2} \left(\rho_{ABC}^{\left(x\right)}\right)
\end{align}

\section{Examples}
\label{Example_channels}
\subsection{Unitary Channel}
 Consider any unitary channel $\mathcal{U} _{A \rightarrow A}$ defined as $\mathcal{U}\left(\rho\right) = U \rho U^{\dag}$. Then if a game is described by $\rho_{ABX}$ where $\text{dim}\left(B\right) \leq \text{dim}\left(A\right)$, the reward is as follows:
\begin{align}
R_{\rho_{ABX}}\left(\mathcal{U}\right) &= \max_{\mathcal{E}} \left( \sum_{x} p_{x} \Big\Vert U \mathcal{E}_{A}\left(\rho_{AB}^{\left(x\right)}\right) U^{\dag} \Big\Vert_{x} \right) \\
&= \max_{\mathcal{E}} \left( \sum_{x} p_{x} \Big\Vert \mathcal{E}_{A}\left(\rho_{AB}^{\left(x\right)}\right) \Big\Vert_{x} \right) \\ 
\end{align}
Consider the game where $\rho_{AB}^{(x)} = \phi_{AB}^{+}$ for all $x$ and where $X$ is drawn according to distribution $\mathbf{p}$. Then 
\begin{align}
    R_{\rho_{ABX}}(\mathcal{U}) &= \max_{\mathcal{E}} \left( \sum_{x} p_{x} \Big\Vert \mathcal{E}_{A}\left(\phi_{AB}^{+}\right) \Big\Vert_{x} \right) \\ 
    & \leq  \sum_{x} p_{x} \Big\Vert \mathcal{I}_{A}\left(\phi_{AB}^{+} \right) \Big\Vert_{x} \\ 
    & =  \sum_{x} p_{x} \\
    & =1  \\ 
\end{align}

Clearly, for the game where $\rho_{AB}^{(x)} = \ket{0}\!\!\bra{0}$ we likewise have
\begin{align}
    R_{\rho_{ABX}}(\mathcal{U}) &= \max_{\mathcal{E}} \left( \sum_{x} p_{x} \left\Vert \mathcal{E}_{A}\left(\ket{0}\!\!\bra{0} \right) \right\Vert_{x} \right) \\ 
    & \leq  \sum_{x} p_{x} \Big\Vert \ket{0}\!\!\bra{0} \Big\Vert_{x} \\ 
    & =  \sum_{x} p_{x} \\
    & =1  \\ 
\end{align}

\subsection{Depolarizing Channel}

First, we consider the reward for $\mathcal{D}_{\gamma}$ in the game where $\rho_{AB}^{(x)} = \phi_{AB}^{+}$ for all $x$:
\begin{align}
    R_{\rho_{ABX}}(& \mathcal{D}_{\gamma})= \max_{\mathcal{E}} \left( \sum_{x} p_{x}  \Big\Vert  \mathcal{D}_{\gamma} \circ \mathcal{E} \left(\phi_{AB}^{+}\right) \Big\Vert_{\left(x\right)} \right) \\
    &= \max_{\mathcal{E}} \left( \sum_{x} p_{x}  \left\Vert (1-\gamma)  \mathcal{E} \left(\phi_{AB}^{+}\right) + \frac{\gamma I_{4}}{4}\right\Vert_{\left(x\right)} \right) \\
    &=  \sum_{x} p_{x}  \left\Vert \left(1-\gamma\right) \phi_{AB}^{+} + \frac{\gamma I_{4}}{4}\right\Vert_{\left(x\right)} \\
    &=  (1-\gamma) + \frac{\gamma}{4}\sum_{x=1}^{4} x p_{x} \\
\end{align}
Likewise, if $\rho_{AB}^{(x)} = \ket{0}\!\!\bra{0}_{A}$ for all $x$, we find:
\begin{align}
    R_{\rho_{ABX}}(& \mathcal{D}_{\gamma})= \max_{\mathcal{E}} \left( \sum_{x} p_{x}  \left\Vert  \mathcal{D}_{\gamma} \circ \mathcal{E} \left(\ket{0}\!\!\bra{0}\right) \right\Vert_{\left(x\right)} \right) \\
    &= \max_{\mathcal{E}} \left( \sum_{x} p_{x}  \left\Vert \left(1-\gamma \right)  \mathcal{E} \left(\ket{0}\!\!\bra{0}\right) + \frac{\gamma I_{2}}{2}\right\Vert_{\left(x\right)} \right) \\
    &=  \sum_{x} p_{x}  \left\Vert (1-\gamma) \ket{0}\!\!\bra{0} + \frac{\gamma I_{2}}{2}\right\Vert_{\left(x\right)} \\
    &=  (1-\gamma) + \frac{\gamma}{2}(p_{1} + 2p_{2}) \\
    &= (1-\gamma) + \gamma(1- \frac{p_{1}}{2})
\end{align}

\subsection{Amplitude Damping Channel}
Recall that the amplitude damping channel is defined as
\begin{align}
    \mathcal{A}_{\gamma}\left(\rho\right) & \triangleq \begin{pmatrix}
    \rho_{00} + \gamma \rho_{11} & \sqrt{1-\gamma} \rho_{01} \\
    \sqrt{1-\gamma} \rho_{10} & \left(1-\gamma\right) \rho_{11}.
    \end{pmatrix}
\end{align}
Then for the game where $\rho_{AB}^{(x)} = \ket{0}\!\!\bra{0}$ for all $x$, we have
\begin{align}
   R_{\rho_{ABX}}( \mathcal{A}_{\gamma}) &= \max_{\mathcal{E}} \left( \sum_{x} p_{x}  \left\Vert  \mathcal{A}_{\gamma} \circ \mathcal{E} (\ket{0}\!\!\bra{0}) \right\Vert_{\left(x\right)} \right) \\
   &= \sum_{x} p_{x}  \left\Vert  \mathcal{A}_{\gamma} \left(\ket{0}\!\!\bra{0}\right) \right\Vert_{\left(x\right)}  \\
   &= \sum_{x} p_{x}  \left\Vert  \ket{0}\!\!\bra{0} \right\Vert_{\left(x\right)} ) \\
   &= 1
\end{align}
Likewise if $\rho_{AB}^{(x)} = \phi_{AB}^{+}$ for all $x$, we have \\
\begin{align}
   & R_{\rho_{ABX}}( \mathcal{A}_{\gamma}) = \max_{\mathcal{E}} \left( \sum_{x} p_{x}  \left\Vert  \mathcal{A}_{\gamma} \circ \mathcal{E} \left(\phi_{AB}^{+}\right) \right\Vert_{\left(x\right)} \right) \\
   &= \max_{\mathcal{E}} \left( \sum_{x} p_{x}  \left\Vert \sum_{y, z =0}^{1}  \mathcal{A}_{\gamma}\! \circ \mathcal{E} \left(\ket{y}\!\!\bra{z}\right) \otimes \ket{y}\!\!\bra{z} \right\Vert_{\left(x\right)} \right) \\
   &= \max_{\phi} \left( \sum_{x} p_{x}  \left\Vert \sum_{y, z \in\{\phi, \phi^{\perp}\}}  \!\!\! \mathcal{A}_{\gamma}\left(\ket{y}\!\!\bra{z}\right) \otimes \ket{y}\!\!\bra{z} \right\Vert_{\left(x\right)} \right) \\
   &=  \sum_{x} p_{x}  \left\Vert \sum_{y, z =0}^{1} \! \mathcal{A}_{\gamma}(\ket{y}\!\!\bra{z}) \otimes \ket{y}\!\!\bra{z} \right\Vert_{\left(x\right)}  \\
   &=  p_{1}(1-\frac{\gamma}{2}) +(1-p_{1}) \\
   &=  1-p_{1}\frac{\gamma}{2} \\
\end{align}

\subsection{Dephasing Channel:}
The dephasing channel $\mathcal{F}_{\gamma}$ can be written as
\begin{align}
    \mathcal{F}_{\gamma}\left(\rho\right) = \left(1-\gamma\right) \rho + \gamma I_{\text{CL}}\left(\rho\right)
\end{align}
where $I_{\text{CL}}$ is the classical identity channel (completely dephasing channel.) Evidently for $\rho_{AB}^{(x)} = \ket{0}\!\!\bra{0}$ for all $x$ we have
\begin{align}
    R_{\rho_{ABX}}(\mathcal{F}_{\gamma}) &= \max_{\mathcal{E}} \left( \sum_{x} p_{x}  \left\Vert  \mathcal{F}_{\gamma} \circ \mathcal{E} \left(\ket{0}\!\!\bra{0}\right) \right\Vert_{\left(x\right)} \right) \\
    &=  \sum_{x} p_{x}  \left\Vert  \mathcal{F}_{\gamma} \left(\ket{0}\!\!\bra{0}\right) \right\Vert_{\left(x\right)} \\
    &=  \sum_{x} p_{x}  \Big\Vert  \ket{0}\!\!\bra{0} \Big\Vert_{\left(x\right)} \\
    &=1
\end{align}

If $\rho_{AB}^{\left(x\right)} = \phi_{AB}^{+}$ for all $x$, 
\begin{align}
    &R_{\rho_{ABX}}( \mathcal{F}_{\gamma})= \max_{\mathcal{E}} \left( \sum_{x} p_{x}  \left\Vert  \mathcal{F}_{\gamma} \circ \mathcal{E} \left(\phi_{AB}^{+}\right) \right\Vert_{\left(x\right)} \right) \\
    &= \max_{\mathcal{E}} \left( \sum_{x} p_{x}  \left\Vert (1-\gamma)  \mathcal{E} (\phi_{AB}^{+}) + \frac{\gamma \ket{0}\!\!\bra{0} \otimes I_{2}}{2}\right\Vert_{\left(x\right)} \right) \\
    &=  \sum_{x} p_{x}  \left\Vert (1-\gamma) \phi_{AB}^{+} + \frac{\gamma \ket{0}\!\!\bra{0} \otimes I_{2}}{2}\right\Vert_{\left(x\right)} \\
    &=  (1-\gamma) + \frac{\gamma}{2}(p_{1}+2(1-p_{1})) \\
    &= 1-  \frac{\gamma}{2}p_{1} \\
\end{align}

Results for the classical identity channel follow immediately by setting $\gamma =1$

\subsection{Projective measurements}

Finally, we consider the case where $\mathcal{N}_{\hat{\Pi}}$ is a channel which implements the projective qubit measurement $\hat{\Pi}= \{\Pi_{0}, \Pi_{1}\}$ s.t. 
\begin{equation}
    \mathcal{N}_{\hat{\Pi}}\left(\rho\right)=\operatorname{Tr}\big[\Pi_{0} \rho\big] \times \Pi_{0} +\operatorname{Tr}\big[\Pi_{1} \rho\big] \times \Pi_{1} 
\end{equation}
where $\Pi_{0}$ and $\Pi_{1}$ are both rank-one measurement elements.

We now calculate the reward for the game where $\rho_{AB}^{(x)} = \ket{0}\!\!\bra{0}$ for all $x$
\begin{align}
    R_{\rho}\left(\mathcal{N}_{\hat{\Pi}}\right) &=\max_{\mathcal{E}}\left(\sum_{x=1}^{2}p_{x}
    \left\|\sum_{i=1}^{2} \text{Tr}\left[ \Pi_{i} \mathcal{E}\left(\ket{0}\!\!\bra{0}\right) \right] \Pi_{i} \right\|_{(x)} \right)\\
    &= \sum_{x=1}^{2}p_{x}\left\|\sum_{i=1}^{2} \text{Tr}\left[ \Pi_{i} \Pi_{0} \right] \times \Pi_{i} \right\|_{(x)} \\
    &= \sum_{x=1}^{2}p_{x}\left\| \Pi_{0} \right\|_{(x)} \\
    &= 1
\end{align}

Finally, we calculate the reward for the game where $\rho_{AB}^{(x)} = \phi_{AB}^{+}$ for all $x$
\begin{align}
    R_{\rho}\left(\mathcal{N}_{\hat{\Pi}}\right) &=\max_{\mathcal{E}}\left(\sum_{x=1}^{2}p_{x}
    \left\|\sum_{i=1}^{2} \text{Tr}\left[ \Pi_{i} \mathcal{E}(\phi_{AB}^{+}) \right]  \Pi_{i} \right\|_{(x)}\right) \\
    &= \sum_{x=1}^{2}p_{x}\left\|\sum_{i=1}^{2} \text{Tr}\left[ \Pi_{i} \left( \Pi_{0} \otimes \frac{I}{2} \right) \right] \Pi_{i} \right\|_{(x)} \\
    &= \sum_{x=1}^{2}p_{x}\left\|\sum_{i=1}^{2}  \Pi_{0} \otimes \frac{I}{2}  \right\|_{(x)} \\
    &= \frac{1}{2}p_{1} + (1-p_{1})
\end{align}
where in the second line we note that the optimal choice of $\mathcal{E}$ is the replacement channel which always prepares output state $\Pi_{0}$ (or alternatively, which always prepares output state $\Pi_{1}$).

\section{Monotonicity Proof}
\label{appendixE}
Suppose that $\rho^{*}$ is the state which satisfies
\begin{align}
    \rho^{*} & \triangleq \underset{\rho}{\text{argmax}} \left( \operatorname{Tr}\Big[ \mathcal{M}\left(\rho\right)^{2} \Big] \right) 
\end{align}
Consider diagonal matrices U, V. Through Holder's inequality, we have:
\begin{align}
    \operatorname{Tr}[UV] & \leq \sqrt{\operatorname{Tr}[U^{2}]} \sqrt{\operatorname{Tr}[V^{2}]}
\end{align}
It follows that if $\operatorname{Tr}[UV] \geq \operatorname{Tr}[U^{2}]$, then
\begin{align}
   \operatorname{Tr}[V^{2}] \geq \operatorname{Tr}[UV] \geq \operatorname{Tr}[U^{2}]
\end{align}
Now consider the game where $\rho_{AB}^{\left(x\right)} = \ket{0}\!\!\bra{0}_{A}$ for all $x$ (namely, the system $B$ is trivial), and where $p_{x}$ satisfy 
\begin{align}
    \sum_{x=j}^{\ell} p_{x} = \alpha \lambda^{\downarrow}_{j}\left(\mathcal{M}\left(\rho^{*}\right)\right)
\end{align}
where $\alpha$ is some positive normalisation constant. Then the reward for $\mathcal{M}$ is simply 
\begin{align}
    R\left(\mathcal{M}\right) &=  \max_{\mathcal{E}} \sum_{x=1}^{\ell} p_{x} \sum_{j=1}^{x} \lambda_{j}^{\downarrow}\left( \mathcal{M}\left( \mathcal{E}\left(\ket{0}\!\!\bra{0}\right) \right) \right)
     \\
      &= \max_{\rho} \sum_{x} p_{x} \sum_{j=1}^{x} \lambda_{j}^{\downarrow}\left( \mathcal{M}\left( \rho\right) \right)  
     \\
     &= \max_{\rho} \sum_{j=1}^{\ell} \sum_{x=j}^{\ell} p_{x} \lambda_{j}^{\downarrow}\left( \mathcal{M}\left( \rho\right) \right)
       \\
     &= \alpha \max_{\rho} \sum_{j} \lambda_{j}^{\downarrow}\left( \mathcal{M}\left(\rho\right) \right) \lambda_{j}^{\downarrow}\left( \mathcal{M}\left(\rho^{*}\right) \right)   \\
    &= \alpha \sum_{j} \lambda_{j}^{\downarrow}\left( \mathcal{M}\left(\rho^{*}\right) \right) \lambda_{j}^{\downarrow}\left( \mathcal{M}\left(\rho^{*}\right) \right)   \\
    &= \alpha \operatorname{Tr}\left[\mathcal{M}\left(\rho^{*}\right)^{2}] \right]
\end{align}
Evidently, since $\mathcal{M} \precsim \mathcal{N}$, 
\begin{align}
    R\left(\mathcal{N}\right) &= \alpha \max_{\rho} \sum_{j} \lambda_{j}^{\downarrow}\left( \mathcal{N}\left(\rho\right) \right) \lambda_{j}^{\downarrow}\left( \mathcal{M}\left(\rho^{*}\right) \right) \\
     &= \alpha \sum_{j} \lambda_{j}^{\downarrow}\left( \mathcal{N}\left(\sigma\right) \right) \lambda_{j}^{\downarrow}\left( \mathcal{M}\left(\rho^{*}\right) \right)   \\
    & \geq \alpha \operatorname{Tr}\left[\mathcal{M}\left(\rho^{*}\right)^{2}\right]  
\end{align}
where in the above we denote the state which attains the optimum for $\mathcal{N}$ as  $\sigma$. Evidently, upon rewriting the above we have
\begin{align}
\operatorname{Tr}\Big[ \text{diag}\left(\mathcal{N}\left(\sigma\right)\right) \text{diag}\left(\mathcal{M}\left(\rho^{*}\right)\right) \Big] \geq \operatorname{Tr}\Big[ \mathcal{M}^{*}\left(\rho\right) \Big]
\end{align}
This implies  
\begin{align}
    \operatorname{Tr}\left[\mathcal{N}\left(\sigma\right)^{2}\right] &\geq \operatorname{Tr}\left[\mathcal{M}\left(\rho^{*}\right)^{2}\right] \\
    & \triangleq \sup_{\rho}\left(\operatorname{Tr} \left[ \mathcal{M}\left(\rho\right)^{2} \right] \right)  
\end{align}
and the statement immediately follows. 



\vspace{.2in}

\onecolumngrid

\end{document}